\let\accentvec\vec  
\documentclass[number,sort&compress]{elsarticlesilent}
\let\vec\accentvec 

\usepackage{amsthm}
\usepackage{xspace}
\usepackage{amssymb}
\usepackage{graphicx} 	
\usepackage[shortlabels]{enumitem}
\usepackage{url}
\usepackage{amsmath}
\usepackage{amssymb}

\usepackage{multirow}
\usepackage{booktabs}

\usepackage{float}
\usepackage{algorithm}
\usepackage[noend]{algorithmic}

\renewcommand{\algorithmicrequire}{\textbf{Input:}}
\renewcommand{\algorithmicensure}{\textbf{Output:}} 

\usepackage[tight]{subfigure} 

\usepackage{tikz}
\usetikzlibrary{shapes}
\usetikzlibrary{calc}
\usetikzlibrary{decorations.pathreplacing}
\tikzset{node/.style={
circle,
very thick,
draw=black!90!white!90, 
top color=white, 
bottom color=red!50!black!20, 
inner sep=0,
minimum size=0.5cm,
align=center
}}

\tikzset{smallNode/.style={
circle,
thick,
draw=black!90!white!90, 
top color=white, 
bottom color=red!50!black!20, 
inner sep=0,
minimum size=0.2cm,
align=center
}}

\newcommand{\inducedPathPicture}{
\begin{figure}[t]
\centering

\small

\subfigure[Input instance $(G_1, v_1, v_4)$ of \HamSTpath.]
{
\begin{tikzpicture}[thick,>=stealth]

\foreach \i \j \k in {1/1/0,2/2/1,3/3/0,4/4/0}
	\node (a-\i) at (\j, \k) [node] {$v_{\i}$};

\foreach \i \j in {1/2,1/3,2/3,3/4}
	\draw (a-\i) -- (a-\j);
\end{tikzpicture}
}
\subfigure[Output instance of \InducedPathByVC.]{
\begin{tikzpicture}[thick,>=stealth]

\coordinate (iw) at (6cm,0cm);
\foreach \i in {1,...,4}
{
	\node (v-\i) at ($\i*(1cm, 0cm) + (0cm, 0cm)$) [node] {$v^*_{\i}$};
	
	\node (z-\i) at ($\i*(1cm, 0cm) + (0cm, 2.5cm)$) [node] {$z_{\i}$};
}

\coordinate (mid) at ($0.5*(v-1) + 0.5*(v-4)$);

\foreach \j \h \i in {1/2/1,1/3/2,1/4/3,2/3/4,2/4/5,3/4/6}
{
	\node (e-\j-\h) at ($\i*(1cm,0cm) + (mid) + (-3.5cm, 1cm)$) [node] {$e_{\j,\h}$};
	\draw (e-\j-\h) -- (v-\j);
	\draw (e-\j-\h) -- (v-\h);
}

\coordinate (left) at ($(e-1-2) + (-1cm, 0.5cm)$);
\coordinate (right) at ($(e-3-4) + (+1cm, 0.5cm)$);

\path ($(left)$) node (yB) [node] {$y_B$}
	+(0, -2cm) node (xB) [node]  {$x_B$}
	++(-1cm, 0cm) node (yA) [node] {$y_A$}
	+(0, -2cm) node (xA) [node]  {$x_A$}
	++(8cm, 0cm) node (yC) [node] {$y_C$}
	+(0, -2cm) node (xC) [node]  {$x_C$};

\draw [black,decorate,decoration={brace,amplitude=5pt}]
   ($(xA.south) + (-0.5cm, 0)$) -- ($(yA.north) + (-0.5cm, 0)$);

\node [rotate=90] at ($0.5*(xA) + 0.5*(yA) + (-1cm, -0.1cm)$) {$n^3$ vertices};

\foreach \i in {A,B,C}
	\draw [dotted] (x\i) -- (y\i);

\foreach \i in {A,B}
	\draw (y\i) -- (z-1);
	
\foreach \i in {A,B}
	\foreach \j in {2,...,4}
		\draw [very thin] (y\i) -- (z-\j);
	
\draw (xB) -- (v-1);
\draw (xC) -- (v-4);

\draw [dashed] ($(yA) + (-1cm, 0.5cm)$) -- ($(yC) + (1cm, 0.5cm)$);

\foreach \i \j \h in {1/1/4,1/2/4}
	\draw (z-\i) -- (e-\j-\h);

\end{tikzpicture}
}
\caption{ \label{inducedPathPicture} An example of the lower-bound construction of \thmref{theorem:inducedPathByVCNoPoly}. (a) The first input instance. (b) The graph~$G^*$ which is the result of cross-composing four inputs on~$n=4$ vertices each. Note that the composition algorithm only builds the graph~$G^*$ when~$n \geq 9$, but this picture gives the correct intuition for the construction. Edges between~$\{z_2, z_3, z_4\}$ and~$\{e_{j,h} \mid \{j,h\} \in \binom{[4]}{2}\}$ have been omitted for readability. The vertices below the horizontal dashed line form the vertex cover~$X^*$.
}
\end{figure}
}

\graphicspath{{./images/}}

\usepackage[pdftex]{hyperref}

\usepackage{doi}

\newcommand{\problemdef}[3]
{
\begin{quote}
\textsc{#1}\\
\textbf{Input:} #2\\
\textbf{Question:} #3
\end{quote}
}

\newcommand{\parproblemdef}[4]
{
\begin{quote}
\textsc{#1}\\
\textbf{Input:} #2\\
\textbf{Parameter:} #3\\
\textbf{Question:} #4
\end{quote}
}

\newcommand{\myqed}{}

\newcommand{\yes}[0]{\textsc{yes}\xspace}
\newcommand{\no}[0]{\textsc{no}\xspace}

\newcommand{\orsymb}[0]{\textsc{or}\xspace}

\newcommand{\poly}[0]{\mathrm{poly}}
\newcommand{\containment}[0]{NP~$\subseteq$~coNP$/$poly\xspace}
\newcommand{\ncontainment}[0]{NP~$\not \subseteq$~coNP$/$poly\xspace}
\newcommand{\F}[0]{\ensuremath{\mathcal{F}}\xspace}
\renewcommand{\S}[0]{\ensuremath{\mathcal{S}}\xspace}

\newcommand{\Reduce}[0]{\textsc{Reduce}\xspace}

\newcommand{\vc}[0]{\mathop{\mathrm{\textsc{vc}}}}

\newcommand{\Oh}[0]{\ensuremath{\mathcal{O}}\xspace}

\newlength{\baseImageHeight}
\newcommand{\XTC}[0]{\textsc{Exact Cover by 3-Sets}\xspace}
\newcommand{\pTwoSplitIS}[0]{\textsc{Independent Set on $P_2$-Split Graphs}\xspace}
\newcommand{\OrIndependentSet}[0]{\orsymb-\textsc{Independent Set}\xspace}
\newcommand{\BipartiteBiclique}[0]{\textsc{Balanced Biclique in Bipartite Graphs}\xspace}
\newcommand{\BipartiteInducedMatching}[0]{\textsc{Maximum Induced Matching in Bipartite Graphs}\xspace}
\newcommand{\BipartiteRegularPerfectCodeByTK}[0]{\textsc{Bipartite Regular Perfect Code ($|T| + k$)}\xspace}
\newcommand{\HMinorTestByVCH}[0]{\textsc{$H$-Minor Test ($\vc + |V(H)|$)}\xspace}
\newcommand{\InducedPsiTestByVC}[0]{\textsc{Induced $\Psi_{s,t}$-Subgraph Test ($\vc$)}\xspace}
\newcommand{\HInducedSubgraphTestByVCH}[0]{\textsc{Induced $H$-Subgraph Test ($\vc + |V(H)|$)}\xspace}
\newcommand{\InducedMatchingByVC}[0]{\textsc{Maximum Induced Matching ($\vc$)}\xspace}
\newcommand{\DistanceHereditaryDeletionVC}[0]{\textsc{Distance Hereditary Deletion ($\vc$)}\xspace}
\newcommand{\InducedPathByVC}[0]{\textsc{Long Induced Path ($\vc$)}\xspace}
\newcommand{\ConnectedDominatingSetVC}[0]{\textsc{Connected Dominating Set ($\vc$)}\xspace}
\newcommand{\InducedPath}[0]{\textsc{Long Induced Path}\xspace}
\newcommand{\LargestInducedPiVC}[0]{\textsc{Largest Induced $\Pi$-Subgraph ($\vc$)}\xspace}
\newcommand{\DeletionToPiFreeVC}[0]{\textsc{Deletion Distance To~$\Pi$-free ($\vc$)}\xspace}

\newcommand{\PartitionPiVC}[0]{\textsc{Partition into~$q$ Disjoint $\Pi$-free Subgraphs ($\vc$)}\xspace}
\newcommand{\PartitionThreeColoringVC}[0]{\textsc{Partition into~$3$ Disjoint $\{K_2\}$-free Subgraphs ($\vc$)}\xspace}

\newcommand{\disjointPaths}[0]{\textsc{Disjoint Paths}\xspace}
\newcommand{\OddCycleTransversal}[0]{\textsc{Odd Cycle Transversal}\xspace}
\newcommand{\EtaTransversal}[0]{\textsc{$\eta$-Transversal}\xspace}
\newcommand{\Planarization}[0]{\textsc{Planarization}\xspace}
\newcommand{\IndependentSet}[0]{\textsc{Independent Set}\xspace}
\newcommand{\VertexCover}[0]{\textsc{Vertex Cover}\xspace}
\newcommand{\IndependentSetByVC}[0]{\textsc{Independent Set ($\vc$)}\xspace}
\newcommand{\VertexCoverByVC}[0]{\textsc{Vertex Cover ($\vc$)}\xspace}
\newcommand{\OddCycleTransversalByVC}[0]{\textsc{Odd Cycle Transversal ($\vc$)}\xspace}
\newcommand{\PlanarizationByVC}[0]{\textsc{Planarization ($\vc$)}\xspace}
\newcommand{\LongPathByVC}[0]{\textsc{Long Path ($\vc$)}\xspace}
\newcommand{\LongCycleByVC}[0]{\textsc{Long Cycle ($\vc$)}\xspace}
\newcommand{\HPackingByVC}[0]{\textsc{$H$-Packing ($\vc$)}\xspace}
\newcommand{\PartitionIntoForestsByVC}[0]{\textsc{Partition into $q$ Forests ($\vc$)}\xspace}
\newcommand{\PartitionIntoDominatingSets}[0]{\textsc{Partition into $q$ Dominating Sets}\xspace}
\newcommand{\PartitionIntoPerfectMatchings}[0]{\textsc{Partition into $q$ Perfect Matchings}\xspace}
\newcommand{\DomaticNumber}[0]{\textsc{Domatic Number}\xspace}
\newcommand{\PartitionIntoHamiltonianSubgraphs}[0]{\textsc{Partition into $q$ Hamiltonian Subgraphs}\xspace}
\newcommand{\PartitionIntoDistanceHereditary}[0]{\textsc{Partition into $q$ Distance-Hereditary Graphs}\xspace}
\newcommand{\PartitionIntoForests}[0]{\textsc{Partition into $q$ Forests}\xspace}

\newcommand{\PartitionIntoPlanarGraphs}[0]{\textsc{Partition into $q$ Planar Graphs}\xspace}
\newcommand{\ChordalDeletionByVC}[0]{\textsc{Chordal Deletion ($\vc$)}\xspace}
\newcommand{\EtaTransversalByVC}[0]{\textsc{$\eta$-Transversal ($\vc$)}\xspace}
\newcommand{\FMinorFreeDeletionByVC}[0]{\textsc{$\F$-Minor-Free Deletion ($\vc$)}\xspace}

\newcommand{\TrianglePacking}[0]{\textsc{Triangle Packing}\xspace}

\newcommand{\IntervalDeletion}[0]{\textsc{Interval Deletion}\xspace}
\newcommand{\orientablegenus}[0]{\textsc{Orientable Genus}\xspace}
\newcommand{\bandwidth}[0]{\textsc{Bandwidth}\xspace}

\newcommand{\ChordalDeletion}[0]{\textsc{Chordal Deletion}\xspace}
\newcommand{\CliqueMinorTestVC}[0]{\textsc{Clique Minor Test ($\vc$)}\xspace}
\newcommand{\TreewidthVC}[0]{\textsc{Treewidth ($\vc$)}\xspace}
\newcommand{\PathwidthVC}[0]{\textsc{Pathwidth ($\vc$)}\xspace}
\newcommand{\StarMinorTestVC}[0]{\textsc{$K_{1,t}$ Minor Test ($\vc$)}\xspace}
\newcommand{\ConstantBicliqueTest}[0]{\textsc{Induced $K_{c,t}$ Subgraph Test ($\vc$)}\xspace}
\newcommand{\VariableBicliqueTest}[0]{\textsc{Induced $K_{s,t}$ Subgraph Test ($\vc$)}\xspace}
\newcommand{\ConstantBicliqueTestClassical}[0]{\textsc{Induced $K_{c,t}$ Subgraph Test}\xspace}
\newcommand{\Clique}[0]{\textsc{Clique}\xspace}
\newcommand{\PerfectDeletion}[0]{\textsc{Perfect Deletion}\xspace}
\newcommand{\HPacking}[0]{\textsc{$H$-Packing}\xspace}
\newcommand{\LongPath}[0]{\textsc{Long Path}\xspace}
\newcommand{\LongCycle}[0]{\textsc{Long Cycle}\xspace}

\newcommand{\HamSTpath}[0]{\textsc{Hamiltonian $s-t$ Path}\xspace}

\newcommand{\eqvr}[0]{\ensuremath{\mathcal{R}}\xspace}

\newcommand{\ChromaticNr}[0]{\textsc{Chromatic Number}\xspace}

\newcommand{\clique}[0]{\textsc{Clique}\xspace}
\newcommand{\Treewidth}[0]{\textsc{Treewidth}\xspace}
\newcommand{\FMinorFreeDeletion}[0]{\textsc{$\F$-Minor-Free Deletion}\xspace}
\newcommand{\Cutwidth}[0]{\textsc{Cutwidth}\xspace}

\newcommand{\qcoloring}[0]{\textsc{$q$-Coloring}\xspace}

\newcommand{\ThreeColoring}[0]{\textsc{$3$-Colo\-ring}\xspace}
\newcommand{\DominatingSet}[0]{\textsc{Dominating Set}\xspace}
\newcommand{\ConnectedDominatingSet}[0]{\textsc{Connected Dominating Set}\xspace}

\newcommand{\WeightedVertexCover}[0]{\textsc{Weighted Vertex Cover}\xspace}

\newtheorem{redrule}{Reduction Rule}

\newtheorem{observation}{Observation}

\newtheorem{theorem}{Theorem}
\newtheorem{proposition}{Proposition}
\newtheorem{corollary}{Corollary}
\newtheorem{lemma}{Lemma}

\newtheorem*{claim}{Claim}

\newenvironment{claimproof}{\begin{proof}}{\end{proof}}

\theoremstyle{definition}

\newtheorem{definition}{Definition}

\newcommand{\sectref}[1]{Section~\ref{#1}}

\newcommand{\defref}[1]{Definition~\ref{#1}}
\newcommand{\lemmaref}[1]{Lemma~\ref{#1}}

\newcommand{\thmref}[1]{Theorem~\ref{#1}}
\newcommand{\corollaryref}[1]{Corollary~\ref{#1}}
\newcommand{\obsref}[1]{Observation~\ref{#1}}
\newcommand{\tableref}[1]{Table~\ref{#1}}
\newcommand{\tablesref}[1]{Tables~\ref{#1}}
\newcommand{\imgref}[1]{Fig.~\ref{#1}}
\newcommand{\algref}[1]{Algorithm~\ref{#1}}
\newcommand{\ruleref}[1]{Rule~\ref{#1}}
\newcommand{\propref}[1]{Property~(\ref{#1})}
\newcommand{\proposref}[1]{Proposition~\ref{#1}}

\newcommand{\condref}[1]{(\ref{#1})}

\begin{document}

\begin{frontmatter}

\title{Preprocessing Subgraph and Minor Problems: When Does a Small Vertex Cover Help?\tnoteref{t1,t2}}

\tnotetext[t1]{This work was supported by the Netherlands Organization for Scientific Research (NWO), project ``KERNELS: Combinatorial Analysis of Data Reduction'', and by the European Research Council (ERC) grant ``Rigorous Theory of Preprocessing'', reference 267959.}

\tnotetext[t2]{An extended abstract of this work appeared at the 7th International Symposium on Parameterized and Exact Computation (IPEC 2012). The present paper contains the full proofs, together with three new theorems (Theorems~\ref{theorem:partitioningVC},~\ref{theorem:inducedpsitestbyvc:nopoly}, and~\ref{theorem:hminortestbyvch:nopoly}).}

\author[uib]{Fedor V.\ Fomin}
\ead{fomin@ii.uib.no}
\address[uib]{Department of Informatics, University of Bergen. PO Box 7803, N-5020, Bergen, Norway. Phone: +47 55 58 40 24. Fax: +47 55 58 41 99.}

\author[uib]{Bart M.\ P.\ Jansen\corref{cor1}}
\ead{bart.jansen@ii.uib.no}

\author[uib]{Micha\l \ Pilipczuk}
\ead{michal.pilipczuk@ii.uib.no}

\cortext[cor1]{Corresponding author}

\begin{abstract}
We prove a number of results around kernelization of problems parameterized by the size of a given vertex cover of the input graph. We provide three sets of simple general conditions characterizing problems admitting kernels of polynomial size. Our characterizations not only give generic explanations for the existence of many known polynomial kernels for problems like \qcoloring, \OddCycleTransversal, \ChordalDeletion, \EtaTransversal, or \LongPath, parameterized by the size of a vertex cover, but also imply new polynomial kernels for problems like \FMinorFreeDeletion, which is to delete at most~$k$ vertices to obtain a graph with no minor from a fixed finite set~$\mathcal{F}$.
 
While our characterization captures many interesting problems, the kernelization complexity landscape of parameterizations by vertex cover is much more involved. We demonstrate this by several results about induced subgraph and minor containment testing, which we find surprising. While it was known that testing for an induced complete subgraph has no polynomial kernel unless \containment, we show that the problem of testing if a graph contains a complete graph on~$t$ vertices as a minor admits a polynomial kernel. On the other hand, it was known that testing for a path on~$t$ vertices as a minor admits a polynomial kernel, but we show that testing for containment of an induced path on~$t$ vertices is unlikely to admit a polynomial kernel.
\end{abstract}
\begin{keyword}
Kernelization Complexity \sep Parameterization by Vertex Cover
\MSC[2010]{05C85,68R10,68Q17,68Q25}
\end{keyword}

\end{frontmatter}

\section{Introduction}
Kernelization is an attempt at providing rigorous  mathematical analysis of preprocessing algorithms. While the initial interest in kernelization was driven mainly by practical applications, it turns out that kernelization provides a deep insight into the nature of fixed-parameter tractability. In the last few years, kernelization has transformed into one of the major research domains of parameterized complexity and many important advances in the area are on kernelization. These advances include general algorithmic findings on problems admitting kernels of polynomial size~\cite{AlonGKSY11,BodlaenderFLPST09,FominLST10,KratschW12} and frameworks for ruling out polynomial kernels under certain complexity-theoretic assumptions~\cite{BodlaenderDFH09,BodlaenderJK11,DellM10,FortnowS11}.

A recent trend in the development of parameterized complexity, and more generally, multivariate analysis~\cite{Niedermeier10}, is the study of the contribution of various structural measurements (i.e., different than just the total input size or expected solution size) to problem complexity. Not surprisingly, the development of kernelization followed this trend, resulting in various kernelization algorithms and complexity lower bounds for different kinds of parameterizations. In parameterized graph algorithms, one of the most important and relevant \emph{complexity measures} of a graph is its treewidth. The algorithmic properties of problems parameterized by treewidth are, by now, well-understood~\cite{BodlaenderK08}. However, from the perspective of kernelization, this complexity measure is too general to obtain positive results: it is known that a multitude of graph problems such as \VertexCover, \DominatingSet, and \ThreeColoring, do not admit polynomial kernels parameterized by the treewidth of the input graphs unless \containment~\cite{BodlaenderDFH09}. This is why parameterization by more restrictive complexity measures, like the minimum size of a feedback vertex set or of a vertex cover, is much more fruitful for kernelization.

In particular, kernelization of graph problems parameterized by the~\emph{vertex cover number}, which is the size of the smallest vertex set meeting all edges, was studied intensively \cite{BodlaenderJK11,BodlaenderJK11b,CyganLPPS11a,CyganLPPS12,DomLS09,JansenK11b}. For example, it has been shown that several graph problems such as \Treewidth~\cite{BodlaenderJK11b}, \EtaTransversal~\cite{CyganLPPS12}, and \ThreeColoring~\cite{JansenK11b}, admit polynomial kernels parameterized by the size of a given vertex cover. On the other hand, under certain complexity-theoretic assumptions it is possible to show that a number of problems including \DominatingSet~\cite{DomLS09}, \Clique~\cite{BodlaenderJK11}, \ChromaticNr~\cite{BodlaenderJK11}, \Cutwidth~\cite{CyganLPPS11a}, and \WeightedVertexCover~\cite{JansenB11}, do not admit polynomial kernels for this parameter. As the vertex cover number is one of the largest structural graph parameters, being at least as large as treewidth and the feedback vertex number, a superpolynomial kernel lower bound for a parameterization by vertex cover immediately rules out the possibility of obtaining polynomial kernels for these smaller parameters (cf.~\cite{FellowsJR12}). Understanding the kernelization complexity for parameterizations by vertex cover forms the first step towards more complex parameterizations. While different kernelization algorithms for various problems parameterized by vertex cover are known, we lack a general characterization of such problems. The main motivation of our work on this paper is the quest for meta-theorems on kernelization algorithms for problems parameterized by vertex cover.

According to Grohe~\cite{Grohe07log}, meta-theorems expose the deep relations between logic and combinatorial structures, which is a fundamental issue of computational complexity. Such theorems also yield a better understanding of the scope of general algorithmic techniques and the limits of tractability. 
The canonical example here is Courcelle's Theorem~\cite{Courcelle90}, which states that all problems expressible in Monadic Second-Order Logic are linear-time solvable on graphs of bounded treewidth. For more restricted parameters such as the vertex cover number, meta-theorems are available with a better dependency on the parameter~\cite{Lampis11a,Ganian12}. 
In kernelization there are meta-theorems showing polynomial kernels for restricted graph families~\cite{BodlaenderFLPST09,FominLST10}. A systematic way to understand the kernelization complexity of parameterizations by vertex cover would therefore be to obtain a meta-theorem capturing a large class of problems admitting polynomial kernels. A natural approach would be to devise a logical formalism capturing the class of problems admitting polynomial kernels parameterized by the vertex cover number. However, such a formalism should to be able to express \VertexCover, which admits polynomial kernel, but not \Clique, which does not~\cite{JansenK12}; it should capture \OddCycleTransversal~\cite{JansenK12} and \LongCycle~\cite{BodlaenderJK12c} but not \DominatingSet~\cite{DomLS09}; and \Treewidth~\cite{BodlaenderJK11b} but not \Cutwidth~\cite{CyganLPPS11a}. This suggests that the constructed logical formalism would be unnecessarily complicated, far from classical logics like Monadic Second-Order Logic or First-Order Logic, and probably also blatantly contrived to the needs. Therefore, we take a different approach: we try to explain the existence of polynomial kernels parameterized by the vertex cover number using new graph-theoretic characteristics.

In this paper, we provide three theorems with general conditions capturing a wide variety of known kernelization results about parameterizations by vertex cover. It has been observed before that reduction rules that identify irrelevant vertices by marking a polynomial number of vertices for each constant-sized subset of the vertex cover, lead to a polynomial kernel for several problems~\cite{BodlaenderJK12c,JansenK11b}. Our first contribution here is to uncover a characteristic of graph problems that explains their amenability to such reduction strategies, and to provide theorems using this characteristic. Roughly speaking, the problem of finding a minimum-size set of vertices that hits all induced subgraphs belonging to some family~$\Pi$ has a polynomial kernel parameterized by vertex cover, if membership in~$\Pi$ is invariant under changing the presence of all but a constant number of (non)edges incident with each vertex (and some technical conditions are met). The problem of finding the largest induced subgraph belonging to~$\Pi$, or of finding a partition of the vertex set into a constant number of sets that each induce~$\Pi$-free subgraphs, have polynomial kernels parameterized by vertex cover under similar conditions. Our general theorems not only capture a wide variety of known results, they also imply results that were not known before. For example, as a corollary of our theorems we establish that the \FMinorFreeDeletion deletion problem (see \sectref{section:compendium} for definitions) has a polynomial kernel for every fixed~$\F$, when parameterized by the size of a vertex cover; it is noteworthy that the degree of the polynomial bounding the kernel size depends only on the maximum degree of graphs in $\F$, and not on their sizes. Our third general theorem, dealing with graph partitioning problems, can be considered as a significant generalization of the polynomial kernel for \qcoloring parameterized by vertex cover~\cite{JansenK11b} since coloring a graph is equivalent to partitioning its vertex set into independent sets. We show that many different graph partitioning problems, such as \PartitionIntoForests~\cite[GT14]{GareyJ79} and \PartitionIntoPlanarGraphs, have polynomial kernels parameterized by vertex cover. Although several partitioning problems were already listed by Garey and Johnson~\cite{GareyJ79}, little was previously known about the their kernelization complexity. Our theorems show that in many cases, effective preprocessing is possible for instances of such problems that have small vertex covers.

After studying the kernelization complexity of vertex-deletion problems, largest induced subgraph problems, and partitioning problems, we turn to two basic graph properties: containing some graph as an induced subgraph or as a minor. It is known that testing for a clique as a subgraph (when the   size of the clique is part of the input) does not admit a polynomial kernel parameterized by vertex cover unless \containment~\cite{BodlaenderJK11}. This is why we find the following result surprising: testing for a clique as a minor admits a polynomial kernel under the chosen parameterization. Driven by our desire to obtain a better understanding of the kernelization complexity of graph problems parameterized by vertex cover, we investigate induced subgraph testing and minor testing for other classes of graphs such as cycles, paths, matchings and stars. It turns out that the kernelization complexity of induced subgraph testing and minor testing is exactly opposite for all these classes. For example, testing for a star minor does not have a polynomial kernel due to its equivalence to \ConnectedDominatingSet~\cite{DomLS09}, but we provide a polynomial kernel for testing the existence of an induced star subgraph by using a guessing step to reduce it to cases that are covered by our general theorems.

The paper is organized as follows. We start by giving preliminaries on parameterized complexity and graph theory in \sectref{section:preliminaries}. We also supply the definitions for the problems that we apply our general theorems to. In \sectref{section:metatheorems} we describe a general reduction scheme, study its properties and use it to derive sufficient conditions for vertex-deletion problems, largest induced subgraph problems, and partitioning problems, to admit polynomial kernels parameterized by vertex cover. In \sectref{section:orderTesting} we investigate the kernelization complexity of induced subgraph versus minor testing for various graph families. A succinct overview of our results is given in \tablesref{table:DeletionToPiTable}, \ref{table:LargestInducedPiTable}, \ref{table:partitionPiTable}, and \ref{table:orderTesting} (pages~\pageref{table:DeletionToPiTable}, \pageref{table:LargestInducedPiTable}, \pageref{table:partitionPiTable}, and \pageref{table:orderTesting}, respectively).

\section{Preliminaries} \label{section:preliminaries}
\subsection{Parameterized Complexity and Kernels}
A parameterized problem~$Q$ is a subset of~$\Sigma^* \times \mathbb{N}$, the second component being the \emph{parameter} which expresses some structural measure of the input. A parameterized problem is (strongly uniformly) \emph{fixed-parameter tractable} if there exists an algorithm to decide whether $(x,k) \in Q$ in time~$f(k)|x|^{\Oh(1)}$ where~$f$ is a computable function. We refer to the textbooks~\cite{DowneyF99,FlumG06,Niedermeier06} for more background on parameterized complexity.

A \emph{kernelization algorithm} (or \emph{kernel}) for a parameterized problem~$Q$ is a polynomial-time algorithm which transforms an instance~$(x,k)$ into an equivalent instance~$(x', k')$ such that~$|x'|, k' \leq f(k)$ for some computable function~$f$, which is the \emph{size} of the kernel. If~$f \in k^{\Oh(1)}$ then this is a \emph{polynomial kernel} (cf.~\cite{GuoN07a,Bodlaender09}).

To prove kernelization lower bounds we frequently use the framework of cross-composition~\cite{BodlaenderJK11}, which builds on earlier work by Bodlaender et al.~\cite{BodlaenderDFH09}, and Fortnow and Santhanam~\cite{FortnowS11}.
\begin{definition}[Polynomial equivalence relation \cite{BodlaenderJK11}] \label{polyEquivalenceRelation}
An equivalence relation~\eqvr on $\Sigma^*$ is called a \emph{polynomial equivalence relation} if the following two conditions hold:
\begin{enumerate}
	\item There is an algorithm that given two strings~$x,y \in \Sigma^*$ decides whether~$x$ and~$y$ belong to the same equivalence class in~$(|x| + |y|)^{\Oh(1)}$ time.
	\item For any finite set~$S \subseteq \Sigma^*$ the equivalence relation~$\eqvr$ partitions the elements of~$S$ into at most~$(\max _{x \in S} |x|)^{\Oh(1)}$ classes.
\end{enumerate}
\end{definition}
\begin{definition}[Cross-composition \cite{BodlaenderJK11}] \label{crossComposition}
Let~$L \subseteq \Sigma^*$ be a set and let~$Q \subseteq \Sigma^* \times \mathbb{N}$ be a parameterized problem. We say that~$L$ \emph{cross-composes} into~$Q$ if there is a polynomial equivalence relation~$\eqvr$ and an algorithm which, given~$r$ strings~$x_1, x_2, \ldots, x_r$ belonging to the same equivalence class of~$\eqvr$, computes an instance~$(x^*,k^*) \in \Sigma^* \times \mathbb{N}$ in time polynomial in~$\sum _{i=1}^r |x_i|$ such that:
\begin{enumerate}
	\item~$(x^*, k^*) \in Q \Leftrightarrow x_i \in L$ for some~$1 \leq i \leq r$,
	\item~$k^*$ is bounded by a polynomial in~$\max _{i=1}^r |x_i|+\log r$.
\end{enumerate}
\end{definition}
\begin{theorem}[\cite{BodlaenderJK11}] \label{crossCompositionNoKernel}
If some set~$L \subseteq \Sigma^*$ is NP-hard under Karp reductions and~$L$ cross-composes into the parameterized problem~$Q$, then there is no polynomial kernel for~$Q$ unless \containment.
\end{theorem}

The set~$\{1, 2, \ldots, n\}$ is abbreviated as~$[n]$. If~$X$ is a finite set then~$\binom{X}{n}$ denotes the collection of all subsets of~$X$ which have size exactly~$n$. Similarly we use $\binom{X}{\leq n}$ for the subsets of size at most~$n$ (including~$\emptyset$). When defining cross-compositions we will use a unique $k$-bit binary representation of integers in the range~$[1 \ldots 2^k]$ by mapping the number~$2^k$ to string consisting of~$k$ zeros. We use the normal binary expansion for the smaller numbers.

\subsection{Graphs}
All graphs we consider are finite, simple, and undirected. An undirected graph~$G$ consists of a vertex set~$V(G)$ and a set of edges~$E(G) \subseteq \binom{V(G)}{2}$. A graph property~$\Pi$ is a (possibly infinite) set of graphs. A graph~$H$ is a \emph{subgraph} of graph~$G$, denoted~$H \subseteq G$, if~$V(H) \subseteq V(G)$ and~$E(H) \subseteq E(G)$. For~$X \subseteq V(G)$ the subgraph \emph{induced} by~$X$ is denoted by~$G[X]$. Its vertex set is~$X \cap V(G)$, and its edge set is~$E(G) \cap \binom{X}{2}$. For a vertex subset~$X$ we use~$G - X$ to denote the subgraph of~$G$ induced by~$V(G) \setminus X$. The disjoint union of~$t$ copies of a graph~$G$ is represented by~$t \cdot G$. We say that a graph~$G$ is \emph{vertex-minimal} with respect to~$\Pi$ if~$G \in \Pi$ and for all~$S \subsetneq V(G)$ the graph~$G[S]$ is not contained in~$\Pi$.

The \emph{open neighborhood} of vertex~$v$ in graph~$G$ is the set~$\{u \in V(G) \mid \{u,v\} \in E(G)\}$, and is denoted by~$N_G(v)$. The \emph{closed neighborhood} of~$v$ is~$N_G[v] := N_G(v) \cup \{v\}$. The notation extends naturally to sets of vertices~$S$. The open neighborhood is~$N_G(S) := \bigcup _{v \in S} N_G(v) \setminus S$, whereas the closed neighborhood is~$N_G[S] := \bigcup _{v \in S} N_G(v) \cup S$. The \emph{degree} of a vertex~$v$ in graph~$G$ is~$\deg_G(v) := |N_G(v)|$. The maximum degree of a vertex in~$G$ is denoted by~$\Delta(G)$. \emph{Contracting} an edge~$\{u,v\} \in E(G)$ in graph~$G$ results in the graph~$G'$ obtained from~$G$ by removing vertices~$u$ and~$v$ together with their incident edges, and adding a new vertex~$x$ with~$N_{G'}(x) := N_G(\{u,v\})$.

A (simple) \emph{path} in~$G$ is a sequence of distinct vertices~$(v_0, v_1, \ldots, v_k)$ such that~$\{v_{i-1}, v_i\} \in E(G)$ for~$i \in [k]$. The \emph{length} of the path is the number~$k$ of edges on it. The vertices~$v_0$ and~$v_k$ are the \emph{endpoints} of the path. A (simple) \emph{cycle} is a sequence of vertices~$(v_0, v_1, \ldots, v_k)$ for~$k \geq 3$ such that the elements~$\{v_1, \ldots, v_k\}$ are pairwise distinct and~$v_0 = v_k$, with~$\{v_{i-1}, v_i\} \in E(G)$ for~$i \in [k]$. The length of a cycle is the number of edges on it. A graph is \emph{Hamiltonian} if there is a cycle that meets all its vertices. An \emph{odd cycle} is a cycle of odd length. A \emph{chord} in a cycle is an edge between two vertices that are not successive on the cycle. A cycle is \emph{chordless} if it is of length at least $4$ and has no chords. A graph is \emph{chordal} if it does not contain any chordless cycles; it is \emph{bipartite} if it does not have an odd cycle. A graph is \emph{perfect} if for all its induced subgraphs the chromatic number equals the size of the largest clique. As conjectured a long time ago~\cite{Berge61}, and proved recently~\cite{ChudnowskyRST06}, a graph is perfect if and only if it does not contain any odd hole or odd anti-hole as an induced subgraph.

The complete graph (clique) on~$t$ vertices is denoted~$K_t$, whereas the complete bipartite graph (biclique) with partite sets of sizes~$s$ and~$t$ is denoted~$K_{s,t}$. The path graph on~$t$ vertices is~$P_t$, whereas the cycle graph on~$t$ vertices is~$C_t$. A graph~$G$ is \emph{empty} if~$E(G) = \emptyset$. A vertex~$v$ is \emph{simplicial} in graph~$G$ if~$N_G(v)$ is a clique. A \emph{minor model} of a graph~$H$ in a graph~$G$ is a mapping~$\phi$ from~$V(H)$ to subsets of~$V(G)$ (called \emph{branch sets}) which satisfies the following conditions: (a) $\phi(u) \cap \phi(v) = \emptyset$ for distinct~$u,v \in V(H)$, (b) $G[\phi(v)]$ is connected for~$v \in V(H)$, and (c) there is an edge between a vertex in $\phi(u)$ and a vertex in~$\phi(v)$ for all~$uv \in E(H)$. Graph~$H$ is a \emph{minor} of~$G$ if~$G$ has a minor model of~$H$. It is easy to see that this is equivalent to saying that~$H$ can be made from~$G$ by a (possibly empty) sequence of vertex deletions, edge deletions, and edge contractions.

A \emph{proper $q$-coloring} of a graph~$G$ is a function~$f \colon V(G) \to [q]$ such that adjacent vertices receive different colors. The \emph{chromatic number} of a graph is the smallest integer~$q$ for which it admits a proper $q$-coloring. An \emph{$H$-packing} in~$G$ is a set of vertex-disjoint subgraphs of~$G$, each of which is isomorphic to~$H$. An $H$-packing is \emph{perfect} if the subgraphs cover the entire vertex set. The minimum size of a vertex cover in a graph~$G$ is denoted by~$\vc(G)$. To understand the applications of our general kernelization theorems to concrete problems, we need graph-theoretic concepts such as planarity and treewidth. As we do not need their formal definitions, we refer the reader to the textbook by Diestel~\cite{Diestel10} for further details. The following proposition will be useful in several occasions when applying our general theorems to the \FMinorFreeDeletion problem.

\begin{proposition} \label{proposition:boundedDegreeMinorModel}
If~$G$ contains~$H$ as a minor, then there is a subgraph~$G^* \subseteq G$ containing an $H$-minor such that~$\Delta(G^*) \leq \Delta(H)$ and~$|V(G^*)| \leq |V(H)| + \vc(G^*) \cdot (\Delta(H) + 1)$.
\end{proposition}
\begin{proof}
Let~$G$ be a graph containing a model~$\phi$ of a graph $H$. We show how to find a subgraph~$G^*$ satisfying the claims.

First, for every edge $uv\in E(H)$ mark an arbitrary edge of~$G$ between $\phi(u)$ and~$\phi(v)$. Then, in each branch set~$\phi(v)$ for~$v \in V(H)$ mark the edges of any inclusion-minimal tree~$T_v$ in~$G$ that contains all the vertices incident with edges marked in the first step. Moreover, for each isolated vertex~$v \in V(H)$ mark an arbitrary vertex in~$\phi(v)$. Now obtain~$G^*$ from~$G$ by deleting unmarked edges, and deleting unmarked vertices which are not incident with a marked edge. It is easy to verify that restricting~$\phi$ to~$G^*$ gives an $H$-model in~$G^*$. To see that~$\Delta(G^*) \leq \Delta(H)$, consider a vertex~$v \in V(G^*)$ and partition the edges incident with it into two types: those which were marked to build a tree~$T_u$ in a branch set~$\phi(u)$ for some~$u \in V(H)$, and those which connect two different branch sets. Suppose~$v$ is incident with~$\ell$ edges of the tree~$T_u$. Then~$T_u$ has at least~$\ell$ leaves other than $v$, and all these leaves connect~$\phi(u)$ to different branch sets. Observe that each connection to a different branch set corresponds to a distinct neighbor of~$u$ in~$H$. As~$u$ has at most~$\Delta(H)$ neighbors in~$H$, there are at most~$\Delta(H)$ connections between the branch set~$\phi(u)$ and other branch sets. Since at least~$\ell$ connections are made by leaves of~$T_u$ unequal to~$v$, edges incident with~$v$ can make at most~$\Delta(H) - \ell$ connections to other branch sets. As this accounts for all edges incident with~$v$ in~$G^*$ it follows that~$\deg_{G^*}(v) \leq \ell + (\Delta(H) - \ell)$. As~$v$ was arbitrary this proves~$\Delta(G^*) \leq \Delta(H)$. It remains to prove that~$|V(G^*)|$ is suitably small.

Let~$X \subseteq V(G^*)$ be a minimum vertex cover of~$G^*$. All isolated vertices in~$G^*$ correspond to isolated vertices in~$H$, so there are at most~$|V(H)|$ of them. The remaining vertices of~$G^*$ which do not belong to~$X$, have at least one neighbor in~$X$ (as~$X$ is a vertex cover and the vertices are not isolated). Since each vertex in~$X$ has degree at most~$\Delta(H)$, the total number of vertices in~$G^*$ is at most~$|V(H)| + |X| + \Delta(H) \cdot |X| \leq |V(H)| + |X| (\Delta(H) + 1)$, which proves the claim.
\myqed
\end{proof}

The following fact will be useful at various points in our proofs.

\begin{proposition} \label{proposition:vcpathcycle}
If a graph $G$ contains $P_t$ (resp. $C_t$) as a subgraph, then $\vc(G)\geq \lfloor t/2 \rfloor$ (resp. $\vc(G)\geq \lceil t/2 \rceil$).
\end{proposition}
\begin{proof}
The claim follows from the observations that the vertex cover number of a subgraph of $G$ cannot be larger than the vertex cover number of $G$, and that a path and a cycle on $t$ vertices have vertex cover numbers $\lfloor t/2 \rfloor$ and $\lceil t/2 \rceil$, respectively.
\end{proof}

\subsection{Problem Definitions} \label{section:compendium}
For completeness we provide a definition for the problems that we apply our general theorems to. We define the problems in the order in which they appear in the summary tables.

\subsubsection{Vertex-Deletion Problems}
The vertex-deletion problems in \tableref{table:DeletionToPiTable} (page~\pageref{table:DeletionToPiTable}) are defined as follows.

\parproblemdef
{\VertexCoverByVC}
{A graph~$G$ with a vertex cover~$X$, and an integer~$k \geq 1$.}
{The size~$|X|$ of the vertex cover.}
{Does~$G$ have a vertex cover of size at most~$k$, i.e., is there a set~$S \subseteq V(G)$ of size at most~$k$ such that~$G - S$ is an empty graph?}

\noindent Note that in the preceding problem, the given vertex cover $X$ may be suboptimal. Hence this can be interpreted as asking for the existence of a smaller vertex cover, when given some approximation.

\parproblemdef
{\OddCycleTransversalByVC}
{A graph~$G$ with a vertex cover~$X$, and an integer~$k \geq 1$.}
{The size~$|X|$ of the vertex cover.}
{Is there a set~$S \subseteq V(G)$ of size at most~$k$ such that~$G - S$ is bipartite?}

\parproblemdef
{\ChordalDeletionByVC}
{A graph~$G$ with a vertex cover~$X$, and an integer~$k \geq 1$.}
{The size~$|X|$ of the vertex cover.}
{Is there a set~$S \subseteq V(G)$ of size at most~$k$ such that~$G - S$ does not have chordless cycles?}

For any finite set of graphs~$\F$ we define the following parameterized problem.
\parproblemdef
{\FMinorFreeDeletionByVC}
{A graph~$G$ with a vertex cover~$X$, and an integer~$k \geq 1$.}
{The size~$|X|$ of the vertex cover.}
{Is there a set~$S \subseteq V(G)$ of size at most~$k$ such that~$G - S$ does not contain any graph in~$\F$ as a minor?}

\parproblemdef
{\PlanarizationByVC}
{A graph~$G$ with a vertex cover~$X$, and an integer~$k \geq 1$.}
{The size~$|X|$ of the vertex cover.}
{Is there a set~$S \subseteq V(G)$ of size at most~$k$ such that~$G - S$ is planar?}

\parproblemdef
{\EtaTransversalByVC}
{A graph~$G$ with a vertex cover~$X$, and an integer~$k \geq 1$.}
{The size~$|X|$ of the vertex cover.}
{Is there a set~$S \subseteq V(G)$ of size at most~$k$ such that~$G - S$ has treewidth at most~$\eta$?}

\subsubsection{Subgraph Problems}
The subgraph testing problems in \tableref{table:LargestInducedPiTable} (page~\pageref{table:LargestInducedPiTable}) are defined as follows.

\parproblemdef
{\LongCycleByVC}
{A graph~$G$ with a vertex cover~$X$, and an integer~$k \geq 1$.}
{The size~$|X|$ of the vertex cover.}
{Does~$G$ contain a simple cycle on at least~$k$ vertices?}

The \LongPathByVC problem is defined analogously, by asking for a path on at least~$k$ vertices. For any graph~$H$, we define the following packing problem.
\parproblemdef
{\HPackingByVC}
{A graph~$G$ with a vertex cover~$X$, and an integer~$k \geq 1$.}
{The size~$|X|$ of the vertex cover.}
{Does~$G$ contain at least~$k$ vertex-disjoint subgraphs isomorphic to~$H$?}

\noindent Observe that the well-known \TrianglePacking problem is the special case of the previous problem where~$H := K_3$.

\subsubsection{Partitioning Problems}
The vertex partitioning problems in \tableref{table:partitionPiTable} (page~\pageref{table:partitionPiTable}) are mostly self-explanatory. To preserve space, we only give one example to illustrate the idea.

\parproblemdef
{\PartitionIntoForestsByVC}
{A graph~$G$ with a vertex cover~$X$.}
{The size~$|X|$ of the vertex cover.}
{Is there a partition of the vertex set into~$q$ sets~$S_1 \cup S_2 \cup \ldots \cup S_q$ such that for each~$i \in [q]$ the subgraph of~$G$ induced by~$S_i$ is a forest?}

The value of~$q$ is treated as a constant in the definition. The other partitioning problems in \tableref{table:partitionPiTable} (page~\pageref{table:partitionPiTable}) are defined in the natural way by changing the restriction on the subgraphs induced by the partite sets.

\section{General Kernelization Theorems} \label{section:metatheorems}

\subsection{Characterization by Few Adjacencies}
In this section we introduce a general reduction rule for problems parameterized by vertex cover, and show that the rule preserves the existence of certain kinds of induced subgraphs. The central concept is the following.

\begin{definition} \label{definition:fewAdjacencies}
A graph property~$\Pi$ is \emph{characterized by~$c_\Pi \in \mathbb{N}$ adjacencies} if for all graphs~$G \in \Pi$, for every vertex~$v \in V(G)$, there is a set~$D \subseteq V(G) \setminus \{v\}$ of size at most~$c_\Pi$ such that all graphs~$G'$ which are obtained from~$G$ by adding or removing edges between~$v$ and vertices in~$V(G) \setminus D$, are also contained in~$\Pi$.
\end{definition}

\noindent The following proposition shows that various graph properties are characterized by few adjacencies.

\begin{proposition} \label{proposition:classesCharacterizedByFew}
The following properties are characterized by constantly many adjacencies: (for any fixed finite set~$\F$, graph~$H$, or~$\ell \geq 4$, respectively)
\begin{enumerate}
	\item Having a Hamiltonian path (resp.\ cycle) ($c_\Pi = 2$).\label{characterization:PathOrCycle}
        \item Having an odd cycle ($c_\Pi = 2$).\label{characterization:OddCycle}
	\item Containing~$H \in \F$ as a minor ($c_\Pi = \max_{H \in \F} \Delta(H)$).\label{characterization:FMinor}
	\item Having a perfect $H$-packing ($c_\Pi = \Delta(H)$).\label{characterization:HPacking}
	\item Having a chordless cycle of length at least~$\ell$ ($c_\Pi = \ell - 1$).\label{characterization:chordlessCycle}
\end{enumerate}
\end{proposition}
\begin{proof}
We prove the claims one by one.

\condref{characterization:PathOrCycle} First consider the property of being Hamiltonian. Take a graph~$G$ with a Hamiltonian cycle~$C$, and consider an arbitrary vertex~$v$ in~$G$. Let~$D$ contain the predecessor and successor of~$v$ on the cycle. Then it is easy to see that changing the presence of edges between~$v$ and~$V(G) \setminus D$, preserves the cycle~$C$. Hence by \defref{definition:fewAdjacencies} this proves that the property of Hamiltonicity is characterized by two adjacencies. As the length of the cycle is not affected, the same proof goes for the property of having an odd cycle, i.e., the property \condref{characterization:OddCycle}. The proof for the property of having a Hamiltonian path is similar; for the endpoints we only have to preserve a single adjacency.

\condref{characterization:FMinor} Let~$\F$ be a finite set of graphs. Let~$G$ contain~$H' \in \F$ as a minor, and let~$v \in V(G)$ be an arbitrary vertex. We give a set~$D \subseteq V(G) \setminus \{v\}$ of size at most~$\max _{H \in \F} \Delta(H)$ such that changing the adjacencies between~$v$ and~$V(G) \setminus D$ preserves the fact that~$G$ has an $H'$-minor. By \proposref{proposition:boundedDegreeMinorModel} a subgraph~$G^*$ of~$G$ with maximum degree at most~$\Delta(H')$ exists, which has an $H'$-minor model~$\phi$. If~$v$ is not contained in graph~$G^*$, then changing the presence of edges incident with~$v$ preserves the minor model~$\phi$ in~$G$. If~$v$ is contained in~$G^*$, then pick~$D := N_{G^*}(v)$ which has size at most~$\Delta(H')$ by the degree bound of~$G^*$ guaranteed by the proposition. Changing adjacencies between~$v$ and~$V(G) \setminus D$ preserves the fact that~$G^*$ is a subgraph of~$G$, and therefore preserves the fact that~$G$ has~$H'$ as a minor; this implies membership in~$\Pi$.

\condref{characterization:HPacking} Fix a graph~$H$ and let~$G$ be a graph with a perfect $H$-packing. For an arbitrary vertex~$v \in V(G)$, consider a perfect $H$-packing in~$G$ and let~$G'$ be the subgraph in the packing which contains~$v$. Picking~$D := N_{G'}(v)$ it follows that~$|D| \leq \Delta(H)$. Changing adjacencies between~$v$ and~$V(G) \setminus D$ in~$G$ preserves the perfect $H$-packing we started from, as all edges incident with~$v$ needed to make the subgraph~$G'$ isomorphic to~$H$ are maintained. Hence the graph resulting from such modifications has a perfect $H$-packing and is contained in~$\Pi$.

\condref{characterization:chordlessCycle} Let~$G$ be a graph with a chordless cycle~$C$ of length at least~$\ell$, and let~$v$ be an arbitrary vertex. If~$v$ does not lie on~$C$ then changing the presence of edges incident with~$v$ preserves~$C$ and results in a graph with a chordless cycle of length at least~$\ell$. Suppose therefore that~$v$ lies on~$C$, and label the vertices on~$C$ as~$(v, v_2, \ldots, v_k)$ for some~$k \geq \ell$. Define~$D := \{v_2, \ldots, v_{\ell - 1}\} \cup \{v_k\}$, i.e.,~$D$ contains the predecessor of~$v$ and its~$\ell - 2$ successors. Now let~$G'$ be obtained from~$G$ by changing the adjacency between~$v$ and~$V(G) \setminus D$. We prove that~$G'$ has a chordless cycle of length at least~$\ell$. Let~$i$ be the smallest integer larger than two such that~$v$ is adjacent to~$v_i$ in~$G'$. As we explicitly preserved the edge from~$v$ to~$v_k$, this is well-defined. Because the vertices~$\{v_2, \ldots, v_{\ell-1}\}$ are contained in~$D$ we know that~$i > \ell-1$ because~$C$ is chordless. Since the only edges that were modified when moving from~$G$ to~$G'$ are incident with~$v$, it follows from the choice of~$i$ that~$(v, v_2, v_3, \ldots, v_i)$ is a chordless cycle in~$G'$ of length at least~$\ell$; this completes the proof.
\myqed
\end{proof}

We give some non-examples to aid the intuition. The properties of having chromatic number at least four, of being a cycle, or of \emph{not} being a perfect graph, cannot be characterized by a constant number of adjacencies. To see this for graphs of chromatic number at least four, consider an odd wheel with a rim of length~$t$: this is the graph built from an odd cycle~$C_t$ by adding a new vertex~$x$, the \emph{hub}, that is adjacent to all vertices of the cycle. As an odd cycle requires three colors in a proper coloring, the adjacency of the hub to all other vertices increases the chromatic number to four. Now observe that removing any edge between the hub and the cycle decreases the chromatic number to three, as the two endpoints of that edge can then share the same color. Hence any vertex set~$D$ that preserves the fact that the chromatic number is at least four, when changing adjacencies between~$x$ and vertices not in~$D$, must contain all vertices of the cycle. Consequently, such sets cannot have constant size: having chromatic number at least four is not characterized by a finite number of adjacencies. Similar constructions can be made for the properties of being a cycle, and for imperfectness. 

Before introducing the reduction rule that is based on characterizations by few adjacencies, we prove that the existence of such characterizations is closed under union and intersection.

\begin{proposition} \label{proposition:closure}
Let~$\Pi$ and~$\Pi'$ be graph properties characterized by~$c_{\Pi}$ and~$c_{\Pi'}$ adjacencies, respectively. The following holds:
\begin{enumerate}
	\item The property~$\Pi \cup \Pi'$ is characterized by~$\max (c_{\Pi}, c_{\Pi'})$ adjacencies. \label{closure:union}
	\item The property~$\Pi \cap \Pi'$ is characterized by~$c_{\Pi} + c_{\Pi'}$ adjacencies. \label{closure:intersection}
\end{enumerate}
\end{proposition}
\begin{proof}
We prove the two items separately.

\condref{closure:union} Let~$G$ be a graph in~$\Pi \cup \Pi'$, and let~$v$ be an arbitrary vertex in~$G$. We have to find a set~$D$ of size at most~$\max (c_{\Pi}, c_{\Pi'})$ that satisfies the conditions of \defref{definition:fewAdjacencies} with respect to~$v$. If~$G \in \Pi$ then the characterization of~$\Pi$ by~$c_{\Pi}$ adjacencies guarantees the existence of a set~$D \subseteq V(G) \setminus \{v\}$ of size at most~$c_{\Pi}$ such that changing adjacencies between~$v$ and~$V(G) \setminus D$ preserves membership in~$\Pi$, and hence in the union~$\Pi \cup \Pi'$. If~$G \in \Pi'$ we similarly find a set of size at most~$c_{\Pi'}$ that preserves membership in~$\Pi'$ and therefore in the union. In either case we find a set of size at most~$\max(c_{\Pi}, c_{\Pi'})$ that satisfies the conditions of \defref{definition:fewAdjacencies}, establishing the characterization of~$\Pi \cup \Pi'$.

\condref{closure:intersection} Let~$G$ be a graph in~$\Pi \cap \Pi'$, and let~$v$ be an arbitrary vertex in~$G$. Let~$D \subseteq V(G) \setminus \{v\}$ be a set of size at most~$c_\Pi$ that preserves membership in~$\Pi$, and let~$D' \subseteq V(G) \setminus \{v\}$ be a set of size at most~$c_{\Pi'}$ preserving membership in~$\Pi'$. Now consider~$D^* := D \cup D'$. Changing adjacencies between~$v$ and~$V(G) \setminus D^*$ preserves membership in~$\Pi$ (since~$D^*$ contains~$D$), and preserves membership in~$\Pi'$ (as~$D^*$ contains~$D'$). Hence the set~$D^*$ of size at most~$c_\Pi + c_{\Pi'}$ preserves membership in the intersection~$\Pi \cap \Pi'$, which proves the claim.
\myqed
\end{proof}

The closure property of \proposref{proposition:closure} can be used to quickly establish that a graph class is characterized by a constant number of adjacencies. Note that for a graph class~$\Pi$ that is characterized by few adjacencies, it may be impossible to characterize its complement~$\overline{\Pi}$ in this way. As a concrete example, consider the graphs~$\Pi$ with at least one edge: these are characterized by one adjacency, but it is easy to see that the graphs~$\overline{\Pi}$ without any edges may need arbitrarily many adjacencies to characterize. Also observe that any finite graph property~$\Pi$ is trivially characterized by~$\max _{G \in \Pi} |V(G)| - 1$ adjacencies (for~$G \in \Pi$ and~$v \in V(G)$, choose~$D$ as~$V(G) \setminus \{v\}$). This will be useful to verify the preconditions to the general kernelization theorems.

The single reduction rule that we use to derive our general kernelization theorems, is the \Reduce procedure presented as \algref{algorithm:reduce}. Its utility for kernelization stems from the fact that it efficiently shrinks a graph to a size bounded polynomially in the cardinality of the given vertex cover~$X$.

\begin{algorithm}[t]
\caption{\Reduce$(\mbox{Graph } G, \mbox{vertex cover } X \subseteq V(G), \ell \in \mathbb{N}, c_\Pi \in \mathbb{N})$}
\begin{algorithmic}
\FOREACH{$Y \in \binom{X}{\leq c_\Pi}$ and partition of~$Y$ into~$Y^+ \cup Y^-$}
		\STATE let~$Z$ be the vertices in~$V(G) \setminus X$ adjacent to all of~$Y^+$ and to none of~$Y^-$
		\STATE mark~$\ell$ arbitrary vertices from~$Z$ (if~$|Z| < \ell$ then mark all of them)
\ENDFOR
\STATE delete from~$G$ all unmarked vertices that are not contained in~$X$
\end{algorithmic}
\label{algorithm:reduce}
\end{algorithm}

\begin{observation} \label{observation:reductionEffect}
For every fixed constant~$c_\Pi$, \Reduce$(G, X, \ell, c_\Pi)$ runs in polynomial time and results in a graph on~$\Oh(|X| + \ell \cdot 2^{c_\Pi} \cdot |\binom{X}{\leq c_\Pi}|) = \Oh(|X| + \ell \cdot |X|^{c_\Pi})$ vertices.
\end{observation}

The soundness of the \Reduce procedure for many types of kernelization comes from the following lemma. It shows that for graph properties~$\Pi$ that are characterized by few adjacencies, an application of \Reduce with parameter~$\ell = s + p$ preserves the existence of induced~$\Pi$ subgraphs of size up to~$p$ that avoid any set of size at most~$s$.

\begin{lemma} \label{lemma:reductionPreservesPi}
Let~$\Pi$ be characterized by~$c_\Pi$ adjacencies, and let~$G$ be a graph with vertex cover~$X$. If~$G[P] \in \Pi$ for some~$P \subseteq V(G) \setminus S$ and~$S \subseteq V(G)$, then for any~$\ell \geq |S| + |P|$ the graph~$G'$ resulting from \Reduce$(G, X, \ell, c_\Pi)$ contains~$P' \subseteq V(G') \setminus S$ such that~$G'[P'] \in \Pi$ and~$|P'| = |P|$.
\end{lemma}
\begin{proof}
Assume the conditions in the lemma statement hold, and let~$R \subseteq V(G)$ be the vertices that are removed by the reduction procedure, i.e.,~$R := V(G) \setminus V(G')$. Let $p_1, p_2, \ldots, p_t$ be an arbitrary ordering of $P \cap R$. We inductively create a sequence of sets $P_0, P_1, \ldots, P_t$ with~$P = P_0$ such that (a) $G[P_i] \in \Pi$, (b) $|P_i|=|P|$, (c) $P_i \cap S = \emptyset$, and (d) $P_i \cap R =\{p_{i+1},p_{i+2},\ldots,p_t\}$ for every $i \in \{0,1,\ldots,t\}$. Note that~$P$ satisfies the constraints imposed on~$P_0$, while existence of~$P_t$ proves the lemma. Hence, we only need to show how to construct~$P_i$ from~$P_{i-1}$ for~$i \in [t]$.

Consider graph $G[P_{i-1}]$ and vertex $p_i \in P_{i-1}$. As $G[P_{i-1}]\in \Pi$, \defref{definition:fewAdjacencies} ensures that there exists a set~$D$ of at most~$c_\Pi$ vertices of~$P_{i-1}$ such that arbitrarily changing adjacencies between~$p_i$ and vertices of~$P_{i-1}\setminus D$ in~$G[P_{i-1}]$ preserves membership in~$\Pi$. Let~$D^+ := N_G(p_i) \cap D$ and $D^- := D \setminus D^+$. Since vertex~$p_i$ is contained in~$R$ and was removed by the reduction process, it follows from the deletion procedure that~$p_i \not \in X$ and therefore that~$D^+ \subseteq N_G(p_i) \subseteq X$ since~$X$ is a vertex cover of~$G$. Let~$D^-_X := D^- \cap X$. Observe that~$p_i$ was a candidate for marking for the partition~$(D^+,D^-_X)$ of~$D \cap X$, but as~$p_i \in R$ it was not marked. Hence, there exist~$\ell \geq |S| + |P|$ marked vertices in~$V(G) \setminus X$ adjacent to all of~$D^+$ and none of~$D^-_X$. As~$|P_{i-1}|=|P|$ and~$p_i$ is not marked, we can find a vertex~$p_i' \in V(G) \setminus X$ that does not belong to~$P_{i-1}$ or~$S$, is marked, and has the same neighborhood in~$D \cap X$ as~$p_i$. Since~$X$ is a vertex cover, both~$p_i$ and~$p'_i$ have all their neighbors in~$X$. As~$p'_i$ is not adjacent to any member of~$D^-_X$, it is not adjacent to~$D^-$. Take~$P_i := (P_{i-1}\cup\{p_i'\})\setminus\{p_i\}$. Note that $|P_i|=|P_{i-1}|=|P|$ and $P_i\cap R=\{p_{i+1},p_{i+2},\ldots,p_t\}$. Moreover, a graph isomorphic to $G[P_i]$ can be obtained from $G[P_{i-1}]$ by changing adjacencies between~$p_i$ and vertices of~$P_{i-1}\setminus D$. The only adjacencies that need to be changed are between~$p_i$ and~$N_G(p_i) \triangle N_G(p_i')\subseteq X$ ($\triangle$ denotes symmetric difference), but this set is disjoint with~$D$ and hence the changes preserve membership in~$\Pi$. As~$P_i$ satisfies all induction claims, this completes the proof.
\myqed
\end{proof}

\subsection{Kernelization for Vertex-Deletion Problems}

Let $\Pi$ be a graph property. 
We present a general theorem providing polynomial kernels for vertex-deletion problems of the following form. 
\parproblemdef
{\DeletionToPiFreeVC}
{A graph~$G$ with a vertex cover~$X$, and an integer~$k \geq 1$.}
{The size~$|X|$ of the vertex cover.}
{Is there a set~$S \subseteq V(G)$ of size at most~$k$ such that~$G - S$ does not contain a graph in~$\Pi$ as an induced subgraph?}

\noindent Observe that~$\Pi$ need not be finite or decidable. The condition that a vertex cover is given along with the input is present for technical reasons; to apply the data reduction schemes presented in this paper, one may simply compute a 2-approximate vertex cover and use that as~$X$.

\begin{theorem} \label{theorem:deletionToPiVC}
If~$\Pi$ is a graph property such that:
\begin{enumerate}[(i)]
	\item $\Pi$ is characterized by~$c_\Pi$ adjacencies, \label{property:deletionToPiFree:adjacencies}
	\item every graph in~$\Pi$ contains at least one edge, and  \label{property:deletionToPiFree:oneEdge}
	\item there is a non-decreasing polynomial~$p \colon \mathbb{N} \to \mathbb{N}$ such that all graphs~$G$ that are vertex-minimal with respect to~$\Pi$ satisfy~$|V(G)| \leq p(\vc(G))$, \label{property:deletionToPiFree:vcbound}
\end{enumerate}
then \DeletionToPiFreeVC has a kernel with~$\Oh((x + p(x)) x^{c_\Pi})$ vertices, where~$x := |X|$.
\end{theorem}

\noindent Before proving the theorem, we briefly discuss its preconditions. Let us first show the necessity of \condref{property:deletionToPiFree:oneEdge} by considering the property~$\Pi$ only consisting of the two-vertex graph without an edge. Then a graph~$G$ is a clique if and only if it does not contain the graph in~$\Pi$ as an induced subgraph, and hence a graph~$G$ has a clique of size at least~$k$ if and only if we can delete at most~$|V(G)| - k$ vertices from~$G$ to make it induced-$\Pi$-free. Observe that~$\Pi$ is characterized by a single adjacency and trivially satisfies \condref{property:deletionToPiFree:vcbound} for~$p(n) = 2$. But \Clique parameterized by vertex cover does not admit a polynomial kernel unless \containment~\cite{BodlaenderJK11}, which explains why \condref{property:deletionToPiFree:oneEdge} is necessary.

To justify \condref{property:deletionToPiFree:adjacencies}, consider the class~$\Pi$ containing the odd holes and odd anti-holes (induced cycles of odd length at least five, and their edge-complements). It is easy to verify that this~$\Pi$ satisfies conditions~\condref{property:deletionToPiFree:oneEdge} and~\condref{property:deletionToPiFree:vcbound}. Now observe that~$G$ has vertex-deletion distance at most~$k$ to property~$\Pi$ if and only if~$G$ can be made perfect by~$k$ vertex deletions, and that the kernelization complexity of \PerfectDeletion parameterized by vertex cover is still open. 

The third condition demands that the size of vertex-minimal graphs in~$\Pi$ is bounded polynomially in their vertex cover number. The condition is needed to make the proof go through. Observe that the restriction to a \emph{polynomial} function in the condition is crucial, as the existence of a (possibly exponential) function is trivial. For any graph property~$\Pi$, the existence of a function~$g \colon \mathbb{N} \to \mathbb{N}$ such that all graphs~$G \in \Pi$ have an induced subgraph~$G' \subseteq G$ contained in~$\Pi$ with~$|V(G')| \leq g(\vc(G'))$ is guaranteed by the fact that graphs of bounded vertex cover are well-quasi-ordered by the induced subgraph relation~\cite{FellowsHR12}.\footnote{Given~$\Pi$, let~$\Pi_n$ be the vertex-minimal graphs in~$\Pi$ with vertex cover number exactly~$n$. The well-quasi-ordering ensures that~$\Pi_n$ is finite; choose~$g(n) := \max _{G \in \Pi_n} |V(G)|$.}

Having justified the preconditions to our general theorem, we give its proof.

\begin{table}[t]
	\centering
{
\small 
\begin{tabular}{@{}lll@{}}
\toprule
Problem & Forbidden property~$\Pi$ & $c_\Pi$ \\ \midrule
\VertexCover & $\{K_2\}$ & $1$ \\
\OddCycleTransversal & Graphs containing an odd cycle & $2$ \\
\ChordalDeletion & Graphs with a chordless cycle & $3$ \\
\FMinorFreeDeletion & Graphs with an~$H \in \F$-minor & $\max_{H \in \F} \Delta(H)$ \\
\Planarization & Graphs with a~$K_5$ or~$K_{3,3}$ minor & $4$ \\
\EtaTransversal & Graphs of treewidth~$> \eta$ & $f(\eta)$ \\ \bottomrule
\end{tabular}
	\caption{Problems that admit polynomial kernels when parameterized by the size of a given vertex cover, by applying \thmref{theorem:deletionToPiVC}.}
\label{table:DeletionToPiTable}
}
\end{table}

\begin{proof}[Proof of \thmref{theorem:deletionToPiVC}]
Consider some input instance~$(G,X, k)$. Firstly, observe that if~$k\geq |X|$, then we clearly have a \yes-instance: removal of~$X$ results in an edgeless graph, which is guaranteed not to contain induced subgraphs from~$\Pi$ due to \propref{property:deletionToPiFree:oneEdge}. Therefore, we may assume that~$k<|X|$ as otherwise we output a trivial \yes-instance.

We let~$G'$ be the result of \Reduce$(G, X, k + p(|X|), c_\Pi)$ and return the instance~$(G', X, k)$, which gives the right running time and size bound by \obsref{observation:reductionEffect}. We need to prove that the output instance~$(G',X,k)$ is equivalent to the input instance~$(G,X,k)$. As $G'$ is an induced subgraph of~$G$, it follows that if~$G-S$ does not contain any graph in~$\Pi$, then neither does~$G'-(S\cap V(G'))$. Therefore, if~$(G,X,k)$ is a \yes-instance, then so is~$(G',X,k)$. Assume then, that~$(G',X,k)$ is a \yes-instance and let~$S$ be a subset of vertices with $|S|\leq k$ such that $G'-S$ does not contain any induced subgraph from~$\Pi$. We claim that~$G-S$ does not contain such induced subgraphs either, i.e., that~$S$ is also a feasible solution for the instance~$(G,X,k)$.

Assume for the sake of contradiction that there is a set~$P \subseteq V(G) \setminus S$ such that~$G[P] \in \Pi$. Consider a \emph{minimal} such set~$P$,  which ensures by \propref{property:deletionToPiFree:vcbound} that~$|P| \leq p(\vc(G[P]))$. As~$P \cap X$ is a vertex cover of~$G[P]$, it follows that $|P| \leq p(|P\cap X|) \leq p(|X|)$. As we executed the reduction with parameter~$\ell = k + p(|X|)$, \lemmaref{lemma:reductionPreservesPi} guarantees the existence of a set~$P' \subseteq V(G') \setminus S$ such that~$G'[P'] \in \Pi$. But this shows that the graph~$G' - S$ contains an induced~$\Pi$ subgraph, contradicting the assumption that~$S$ is a solution for~$G'$ and thereby concluding the proof.
\myqed
\end{proof}

\begin{corollary} \label{cor:deletionToPiImplications}
All problems in \tableref{table:DeletionToPiTable} fit into the framework of \thmref{theorem:deletionToPiVC} and hence admit polynomial kernels parameterized by the size of a given vertex cover.
\end{corollary}

\begin{proof}
We consider the problems in the order of \tableref{table:DeletionToPiTable} and show how they fit into the framework.

\VertexCoverByVC. Observe that a graph~$G$ has a vertex cover of size~$\ell$ if and only there is a set~$S \subseteq V(G)$ of size~$\ell$ such that~$G - S$ is an independent set, or equivalently,~$G-S$ does not have~$K_2$ as an induced subgraph. So \VertexCoverByVC is equivalent to \DeletionToPiFreeVC for~$\Pi = \{K_2\}$. Since this~$\Pi$ contains only a single graph of degree one, it is easily seen to be characterized by the single adjacency of one vertex in~$K_2$ to its neighbor (\propref{property:deletionToPiFree:adjacencies}). Obviously all graphs in~$\Pi$ contain at least one edge (\propref{property:deletionToPiFree:oneEdge}), and since~$\Pi$ contains a single graph on two vertices, having a vertex cover of size one, the constant function~$p(n) := 2$ suffices for \propref{property:deletionToPiFree:vcbound}. Hence all preconditions for \thmref{theorem:deletionToPiVC} are satisfied and the problem has a kernel with~$\Oh(|X|^2)$ vertices.

\OddCycleTransversalByVC. A graph~$G$ is bipartite if and only if it does not contain a graph with an odd cycle as an induced subgraph. Hence by letting~$\Pi$ contain all graphs which contain an odd cycle (which is not the same as letting~$\Pi$ be the class of all odd cycles), \OddCycleTransversalByVC is equivalent to \DeletionToPiFreeVC. By \proposref{proposition:classesCharacterizedByFew}, this property~$\Pi$ is characterized by a constant number of adjacencies; the proof of the proposition shows that~$c_\Pi := 2$ suffices. Since all graphs with an odd cycle have at least one edge, the second condition is satisfied as well. For the last condition, consider a vertex-minimal graph~$G$ with an odd cycle; such a graph is Hamiltonian, so it has a cycle on $|V(G)|$ vertices as a subgraph. By Proposition~\ref{proposition:vcpathcycle} we have that $|V(G)| \leq 2\vc(G)$, which proves that~$p(n) := 2n$ suffices for the polynomial in \propref{property:deletionToPiFree:vcbound}. We obtain a kernel with~$\Oh(|X|^3)$ vertices.

\ChordalDeletionByVC. A graph~$G$ is chordal if all its cycles of length at least four have a chord; this can be stated equivalently as saying that it does not contain a graph with a chordless cycle as an induced subgraph. If we take~$\Pi$ to be the class of graphs which have a chordless cycle, we can express \ChordalDeletionByVC as an instantiation of \DeletionToPiFreeVC. The proof of \proposref{proposition:classesCharacterizedByFew} shows the property is characterized by three adjacencies. As all graphs with a chordless cycle contain an edge, the second property is satisfied. Similarly as before, a vertex-minimal graph with a chordless cycle is Hamiltonian and hence~$p(n) := 2n$ suffices for \propref{property:deletionToPiFree:vcbound}. The resulting kernel has~$\Oh(|X|^4)$ vertices.

\FMinorFreeDeletionByVC. If we let~$\Pi$ contain all graphs that contain a member of~$\F$ as a minor, then a graph is $\Pi$-induced-subgraph-free if and only if it is $\F$-minor-free. By \proposref{proposition:classesCharacterizedByFew} this class~$\Pi$ is characterized by~$c_\Pi := \max_{H \in \F} \Delta(H)$ adjacencies, so we satisfy \propref{property:deletionToPiFree:adjacencies}. If~$\F$ contains an empty graph, then \F-minor-free graphs have constant size and the problem is polynomial-time solvable; hence in interesting cases the graphs containing a minor from~$\F$ have at least one edge (\propref{property:deletionToPiFree:oneEdge}). Finally, consider a vertex-minimal graph~$G^*$ which contains a graph~$H \in \F$ as a minor.  By~\proposref{proposition:boundedDegreeMinorModel} we have~$|V(G^*)| \leq |V(H)| + \vc(G^*) \cdot (\Delta(H) + 1)$. As~$\F$ is fixed, the maximum degree and size of graphs in~$\F$ are constants which shows that \propref{property:deletionToPiFree:vcbound} is satisfied, resulting in a kernel with~$\Oh(|X|^{\Delta + 1})$ vertices for~$\Delta := \max_{H \in \F} \Delta(H)$.

\PlanarizationByVC. Since this problem is a special case of~\FMinorFreeDeletionByVC for~$\F := \{K_5, K_{3,3}\}$, and both forbidden minors are nonempty, the proof given above shows that this problem has a kernel with~$\Oh(|X|^5)$ vertices.

\EtaTransversalByVC. Recall that the \EtaTransversal problem asks for a vertex set whose removal results in a graph of treewidth at most~$\eta$. Since treewidth does not increase when taking a minor~\cite[Lemma 16]{Bodlaender98}, the class of graphs of treewidth at most~$\eta$ is closed under minors. By the famous results of Robertson and Seymour~\cite{RobertsonS04}, this implies that for each~$\eta$ there is a finite obstruction set~$\F_\eta$ such that~$G$ has treewidth at most~$\eta$ if and only if~$G$ avoids all graphs in~$\F_\eta$ as a minor. It is easy to see that the minimal obstruction sets~$\F_\eta$ do not contain empty graphs, as empty graphs have treewidth zero and cannot be obstructions to having treewidth~$\eta \geq 0$. Hence we may obtain a polynomial kernel for \EtaTransversalByVC by using the obstruction set~$\F_\eta$ in the more general \FMinorFreeDeletionByVC scheme. The kernel size is~$\Oh(|X|^{\Delta  + 1})$ where~$\Delta := \max _{H \in \F_\eta} \Delta(H)$.
\myqed
\end{proof} 

Using \proposref{proposition:classesCharacterizedByFew} and \proposref{proposition:closure} it is easy to apply \thmref{theorem:deletionToPiVC} to many other vertex-deletion problems. For example, a graph is \emph{distance hereditary} if and only if it excludes the house, gem, domino and holes (chordless cycles of length at least five) as induced subgraphs~\cite[Theorem 10.1.1]{BrandstadtLS99}. (The house, gem and domino are fixed, constant-size graphs~\cite[Chapter 1]{BrandstadtLS99}.) Hence if we take~$\Pi$ to contain these constant graphs, together with the graphs that contain a chordless cycle of length at least five, then a graph is distance hereditary if and only if it is induced $\Pi$-free. Since~$\Pi$ is the union of a finite graph property~$\{\mathrm{house},\mathrm{gem},\mathrm{domino}\}$ with the graphs containing a chordless cycle of length at least five, and both are characterized by a constant number of adjacencies, it follows from \proposref{proposition:closure} that~$\Pi$ is characterized by a constant number of adjacencies. It is easy to verify that the other preconditions to \thmref{theorem:deletionToPiVC} are satisfied as well, which implies a polynomial kernel for \DistanceHereditaryDeletionVC. Using this recipe one can obtain polynomial kernels for a host of vertex-deletion problems, whose corresponding graph classes can be defined by combining the elements of \proposref{proposition:classesCharacterizedByFew} with a finite number of arbitrary forbidden induced subgraphs. We do not list all these possible applications here, but move on to our next general theorem.

\subsection{Kernelization for Largest Induced Subgraph Problems}
In this section we study the following class of problems, which is in some sense dual to the class considered previously. For a graph property $\Pi$, we define 

\parproblemdef
{\LargestInducedPiVC}
{A graph~$G$ with a vertex cover~$X$, and an integer~$k \geq 1$.}
{The size~$|X|$ of the vertex cover.}
{Is there a set~$P \subseteq V(G)$ of size at least~$k$ such that~$G[P] \in \Pi$?}

\noindent The following theorem gives sufficient conditions for the existence of polynomial kernels for such problems.

\begin{theorem} \label{theorem:largestInducedPiVC}
If~$\Pi$ is a graph property such that:
\begin{enumerate}[(i)]
	\item $\Pi$ is characterized by~$c_\Pi$ adjacencies, and \label{property:largestInducedPi:adjacencies}
	\item there is a non-decreasing polynomial~$p \colon \mathbb{N} \to \mathbb{N}$ such that all graphs~$G \in \Pi$ satisfy~$|V(G)| \leq p(\vc(G))$, \label{property:largestInducedPi:vcbound}
\end{enumerate}
then \LargestInducedPiVC has a kernel with~$\Oh(p(x) \cdot x^{c_\Pi})$ vertices, where~$x := |X|$.
\end{theorem}

There is a natural example showing the necessity of the first condition in \thmref{theorem:largestInducedPiVC}. If we take~$\Pi$ as the class of all cliques, then testing whether a graph~$G$ has an induced subgraph in~$\Pi$ on at least~$k$ vertices is equivalent to asking whether~$G$ has a clique of size at least~$k$. Since the vertex count of a complete graph exceeds its vertex cover number by exactly one, the class of cliques satisfies~\condref{property:largestInducedPi:vcbound}. The conditional superpolynomial kernel lower bound for \clique parameterized by vertex cover explains why \condref{property:largestInducedPi:adjacencies} is necessary; the class of cliques is not characterized by any constant number of adjacencies. 

The second condition of \thmref{theorem:largestInducedPiVC} is needed to ensure that the resulting problems have kernels at all. Observe that we do not require the set of graphs~$\Pi$ to be decidable. In the absence of the second condition, we could let~$\Pi$ contain all $i$-vertex graphs for which the $i$-th Turing machine halts on a blank tape. This class is trivially characterized by zero adjacencies, since membership in~$\Pi$ only depends on the number of vertices. If the \LargestInducedPiVC problem for this class~$\Pi$ would have a kernel, then we could decide the Halting problem as follows. To decide whether the $i$-th Turing machine halts, we create the edgeless graph~$G_i$ on~$i$ vertices with an empty vertex cover. By the definition of~$\Pi$, the $i$-th machine halts if and only if~$G_i$ has an induced~$\Pi$ subgraph on~$i$ vertices. Running the supposed kernelization on this instance would yield an equivalent, constant-size instance as the parameter value is zero. We could then decide the problem by looking up the answer in a table for constant-size instances hard-coded into the algorithm, thereby solving the Halting problem. The requirement that the size of the graphs in~$\Pi$ is bounded in terms of their vertex cover number, is therefore entirely natural. We need the dependence to be polynomial in order to obtain our polynomial kernel.

Having justified the preconditions, we present the proof of the theorem.

\begin{proof}[Proof of \thmref{theorem:largestInducedPiVC}]
The kernelization reduces an instance~$(G,X,k)$ by executing \Reduce$(G, X, p(|X|), c_\Pi)$ to obtain a graph~$G'$, and outputs the instance~$(G', X, k)$. By \obsref{observation:reductionEffect} this can be done in polynomial time and results in a graph whose size is appropriately bounded; it remains to prove that the two instances are equivalent. 

Since~$G'$ is an induced subgraph of~$G$, any solution contained in~$G'$ is also contained in~$G$: so if $(G',X,k)$ is a \yes-instance, then~$(G,X,k)$ is as well. Assume then that~$(G,X,k)$ is a \yes-instance and let~$P \subseteq V(G)$ be such that~$G[P]\in \Pi$ and $|P| \geq k$. Clearly,~$X\cap P$ is a vertex cover of~$G[P]$, so~$|P|\leq p(|X \cap P|)\leq p(|X|)$ by \propref{property:largestInducedPi:vcbound}. Since the reduction procedure is executed with a value~$\ell := p(|X|)$ and~$|P| \leq p(|X|)$, by applying \lemmaref{lemma:reductionPreservesPi} with an empty set for~$S$ we find that~$G'$ contains a set~$P'$ of the same size as~$P$ such that~$G'[P] \in \Pi$. This proves that~$(G', X, k)$ is a \yes-instance and shows the correctness of the kernelization.
\myqed
\end{proof}

\begin{table}[t]
	\centering
{
\small
\begin{tabular}{@{}lll@{}}
\toprule
Problem & Desired property~$\Pi$ & $c_\Pi$ \\ \midrule
\LongCycle & Graphs with a Hamiltonian cycle & $2$ \\
\LongPath & Graphs with a Hamiltonian path & $2$ \\
\HPacking & Graphs with a perfect~$H$-packing & $\Delta(H)$\\ \bottomrule
\end{tabular}
	\caption{Problems that admit polynomial kernels when parameterized by the size of a given vertex cover, by applying \thmref{theorem:largestInducedPiVC}.}
\label{table:LargestInducedPiTable}
}
\end{table}

\begin{corollary} \label{cor:largestInducedPiImplications}
All problems in \tableref{table:LargestInducedPiTable} fit into the framework of \thmref{theorem:largestInducedPiVC} and admit polynomial kernels parameterized by the size of a given vertex cover.
\end{corollary}
\begin{proof}
We consider the problems in the order of \tableref{table:LargestInducedPiTable} and show how they fit into the framework.

\LongCycleByVC. Observe that if~$G$ has a cycle on~$k$ vertices~$(v_1, \ldots, v_k)$ then the graph~$G[\{v_1, \ldots, v_k\}]$ is Hamiltonian. So~$G$ has a $k$-cycle if and only if~$G$ has an induced Hamiltonian subgraph on~$k$ vertices. Hence \LongCycleByVC is equivalent to \LargestInducedPiVC by letting~$\Pi$ be the class of Hamiltonian graphs. By \proposref{proposition:classesCharacterizedByFew} this class is characterized by two adjacencies. By Proposition~\ref{proposition:vcpathcycle}, for all Hamiltonian graphs~$G'$ it holds that~$|V(G')| \leq 2|\vc(G')|$. Hence \propref{property:largestInducedPi:vcbound} is satisfied as well and we obtain a kernel with~$\Oh(|X|^3)$ vertices. The proof for \LongPathByVC is analogous.

\HPackingByVC. A graph~$G$ admits an $H$-packing of~$k$ disjoint subgraphs, if and only if~$G$ has an induced subgraph on~$k \cdot |V(H)|$ vertices which admits a perfect $H$-packing. If~$H$ is an empty graph then the answer is trivial: there are~$k$ vertex-disjoint subgraphs isomorphic to~$H$ if and only if the vertex count is at least~$k \cdot |V(H)|$. We can therefore solve the case that~$H$ is an empty graph in polynomial time, and focus on the case that~$H$ is nonempty. Choosing~$\Pi$ as the graphs with a perfect $H$-packing allows us to model the packing problem as an instantiation of \LargestInducedPiVC, by scaling the target value~$k$ by a factor~$|V(H)|$. \proposref{proposition:classesCharacterizedByFew} shows that~$\Pi$ is characterized by~$\Delta(H)$ adjacencies. Let us now prove that the second condition is satisfied for this~$\Pi$, by utilizing the fact that we demand~$H$ to be nonempty. Consider a graph~$G$ with a perfect $H$-packing for a nonempty~$H$, and let~$X$ be a minimum vertex cover of~$G$. Each subgraph in the packing contains at least one edge, so each subgraph in the packing has size~$|V(H)|$ and contains a vertex from~$X$. Hence~$|V(G)| \leq |X| \cdot |V(H)|$, which proves that~$p(n) := n \cdot |V(H)|$ suffices for the polynomial. For fixed~$H$ this results in a kernel with~$\Oh(|X|^{\Delta(H)+1})$ vertices.
\myqed
\end{proof}

\subsection{Kernelization for Graph Partitioning Problems}
Having considered induced subgraph testing and vertex-deletion problems in the previous two sections, we now change our focus to partitioning problems. More concretely, we consider problems that ask for the existence of a partition of the vertex set into a constant number of partite sets such that each partite set induces a subgraph of a desired form. For a graph property $\Pi$, the parameterized problem we study is formally defined as follows.

\parproblemdef
{\PartitionPiVC}
{A graph~$G$ with vertex cover~$X \subseteq V(G)$.}
{The size~$|X|$ of the vertex cover.}
{Is there a partition of the vertex set into~$q$ sets~$S_1 \cup S_2 \cup \ldots \cup S_q$ such that for each~$i \in [q]$ the graph~$G[S_i]$ does not contain a graph in~$\Pi$ as an induced subgraph?}

\noindent Note that the value of~$q$ is treated as a constant in the above definition. To give an example of a problem that can be captured by this template, consider the \ThreeColoring problem which asks whether the graph admits a proper coloring with three colors. Such a coloring is  a partition of its vertex set into three independent sets. Observing that a vertex set  is independent if and only if it induces a subgraph  excluding~$K_2$ as an induced subgraph, we see that \ThreeColoring parameterized by vertex cover can be phrased as \PartitionThreeColoringVC. Further applications will be discussed after establishing a sufficient condition for polynomial kernelizability of the general problem.

The kernelization scheme once again uses the \Reduce routine as its single reduction rule. Before presenting the kernel, we derive a lemma that shows how an application of \Reduce affects instances of partitioning problems. In the following we say that a graph~$G$ can be partitioned into~$q$ disjoint~$\Pi$-free subgraphs if there is a partition of~$V(G)$ into~$S_1 \cup \ldots \cup S_q$ such that for all~$i \in [q]$ the graph~$G[S_i]$ does not contain a member of~$\Pi$ as an induced subgraph.

\begin{lemma} \label{lemma:reductionPreservesPartition}
Let~$\Pi$ be characterized by~$c_\Pi$ adjacencies, and let~$p \colon \mathbb{N} \to \mathbb{N}$ be a non-decreasing polynomial such that all graphs~$G^*$ that are vertex-minimal with respect to~$\Pi$ satisfy~$|V(G^*)| \leq p(\vc(G^*))$. Let~$G$ be a graph with vertex cover~$X$, and let~$G'$ be the graph resulting from \Reduce$(G, X, q \cdot p(|X|), q \cdot c_\Pi)$. If~$G'$ can be partitioned into~$q$ disjoint~$\Pi$-free subgraphs, then such a partition exists for~$G$ as well.
\end{lemma}
\begin{proof}
Assume the conditions in the lemma statement hold, and let~$R \subseteq V(G)$ be the vertices that are removed by the reduction procedure, i.e.,~$R := V(G) \setminus V(G')$. Let~$r_1, r_2, \ldots, r_t$ be an arbitrary ordering of~$R$. Assume that~$\S = (S_1, S_2,\ldots,S_q)$ is a partition of~$V(G')$ such that for each~$i \in [q]$ the graph~$G'[S_i]$ does not contain an induced subgraph from~$\Pi$. We inductively create a sequence of set families $\S_0, \S_1, \ldots, \S_t$ with~$\S = \S_0$ such that $\S_i$ is a partition of~$V(G') \cup \{r_1, \ldots, r_i\}$ into~$q$ sets~$S_i^1, \ldots, S_i^q$, and for all~$j \in [q]$ the graph~$G[S_i^j]$ does not contain a graph in~$\Pi$ as an induced subgraph. Note that~$\S$ satisfies the constraints imposed on~$\S_0$, while existence of~$\S_t$ proves the lemma. Hence, we only need to show how to construct~$\S_i$ from~$\S_{i-1}$ for~$i \in [t]$.

To construct the partition~$\S_i$ out of the partition~$\S_{i-1}$ we will show that there is a partite set~$S_{i-1}^j$ to which vertex~$r_i$ can be added, such that~$G[S_{i-1}^j \cup \{r_i\}]$ does not contain a graph in~$\Pi$. The partition~$\S_i$ is then obtained by replacing~$S_{i-1}^j$ by~$S_{i-1}^j \cup \{r_i\}$ in partition~$\S_{i-1}$. Hence it remains to prove that a suitable partite set exists.

Assume for a contradiction that for all~$j \in [q]$, the graph~$G[S_{i-1}^j \cup \{r_i\}]$ contains an induced~$\Pi$ subgraph. For all~$j \in [q]$ let~$H_j \in \Pi$ be an induced subgraph of~$G[S_{i-1}^j \cup \{r_i\}]$ that is vertex-minimal with respect to~$\Pi$. By the induction hypothesis, each such subgraph in~$\Pi$ must contain~$r_i$.

Since~$\Pi$ is characterized by~$c_\Pi$ adjacencies, it follows that for each~$H_j$ there is a set~$D_j \subseteq V(H_j) \setminus \{r_i\}$ of size at most~$c_\Pi$ such that changing the adjacencies between~$r_i$ and~$V(H_j) \setminus D_j$ in~$H_j$ preserves membership in~$\Pi$. Now consider the union~$D := \bigcup _{j=1}^q D_j$, and let~$D_X := D \cap X$ be its intersection with~$X$.

By the choice of parameters to \Reduce and the fact that $r_i$ was not marked, we know that for the subset~$D_X$ of~$X$ of size at most~$q \cdot c_\Pi$ the procedure marked~$q \cdot p(|X|)$ vertices~$Z_{D_X} \subseteq V(G) \setminus X$ such that all~$z \in Z_{D_X}$ have the same neighborhood into~$D_X$ as~$r_i$, i.e., for which~$N_G(z) \cap D_X = N_G(r_i) \cap D_X$. These vertices~$Z_{D_X}$ were consequently preserved in~$G'$. We will show that there is a vertex~$z^* \in Z_{D_X}$ that is not contained in any forbidden graph~$H_j$ for~$j \in [q]$. To see this, observe first that~$r_i \not \in Z_{D_X}$ since~$r_i$ was removed from the graph by the reduction procedure whereas all vertices in~$Z_{D_X}$ were marked to survive in~$G'$. Since~$X$ is a vertex cover of~$G$, for each~$j \in [q]$ the intersection~$V(H_j) \cap X$ is a vertex cover of~$H_j$. The precondition to the lemma therefore implies that~$|V(H_j)| \leq p(\vc(H_j)) \leq p(|X|)$. The total number of vertices in the union of the graphs~$H_j$ is therefore at most~$q \cdot p(|X|)$. Since all these graphs contain~$r_i$, while~$r_i \not \in Z_{D_X}$, the fact that~$|Z_{D_X}| = q \cdot p(|X|)$ therefore implies that there is indeed a vertex~$z^* \in Z_{D_X}$ that is not contained in any graph~$H_j$ for~$j \in [q]$. 

Let~$j^*$ be the index of the partite set of~$\S_{i-1}$ that contains~$z^*$, such that~$z^* \in S_{i-1}^{j^*}$. We will use the characterization of~$\Pi$ by few adjacencies to show that~$r_i$ can be replaced by~$z^*$ in the forbidden graph~$H_{j^*}$ while preserving membership in~$\Pi$, thereby obtaining the contradiction that~$G[S_{i-1}^{j^*}]$ contains a graph in~$\Pi$. Since neither~$z^*$ nor~$r_i$ is contained in the vertex cover~$X$ by the definition of \Reduce --- it only marks and deletes vertices outside~$X$ --- it follows that~$N_G(z^*) \subseteq X$ and~$N_G(r_i) \subseteq X$. Hence~$N_G(z^*) \cap (D \setminus X) = N_G(r_i) \cap (D \setminus X) = \emptyset$. By choice of~$z^*$ we have that~$N_G(z^*) \cap D_X = N_G(r_i) \cap D_X$. Combining the last two statements shows that~$N_G(z^*) \cap D = N_G(r_i) \cap D$. Hence, starting from the graph~$H_{j^*}$, we can obtain the graph~$G[(V(H_{j^*}) \setminus \{r_i\}) \cup \{z^*\}]$ by changing the label of~$r_i$ to~$z^*$, and changing adjacencies between the resulting~$z^*$ and vertices outside the set~$D$. But since~$D$ contains the set~$D_{j^*}$, which preserves membership of~$H_{j^*}$ in~$\Pi$, this transformation preserves membership in~$\Pi$ and therefore~$G[(V(H_{j^*}) \setminus \{r_i\}) \cup \{z^*\}]$ is contained in~$\Pi$. But this graph is an induced subgraph of~$G[S_{i-1}^{j^*}]$, thereby proving that the partition~$\S_{i-1}$ that we started from is not valid since its $j^*$-th partite set induces a graph containing a member of~$\Pi$. It follows that when we start from a valid partition~$\S_{i-1}$, there is a partite set to which~$r_i$ can be added without creating forbidden subgraphs. This proves the lemma.
\myqed
\end{proof}

Armed with this lemma we state the general kernelization theorem for partitioning problems.

\begin{theorem} \label{theorem:partitioningVC}
If~$\Pi$ is a graph property such that:
\begin{enumerate}[(i)]
	\item $\Pi$ is characterized by~$c_\Pi$ adjacencies, and \label{property:partitioning:adjacencies}
	\item there is a non-decreasing polynomial~$p \colon \mathbb{N} \to \mathbb{N}$ such that all graphs~$G$ that are vertex-minimal with respect to~$\Pi$ satisfy~$|V(G)| \leq p(\vc(G))$, \label{property:partitioning:vcbound}
\end{enumerate}
then \PartitionPiVC has a kernel with~$\Oh(p(x) \cdot x^{q \cdot c_\Pi})$ vertices, where~$x := |X|$.
\end{theorem}
\begin{proof}
The kernelization reduces an instance~$(G,X)$ of \PartitionPiVC by executing \Reduce$(G, X, q \cdot p(|X|), q \cdot c_\Pi)$ to obtain a graph~$G'$, and outputs the instance~$(G', X)$. As before, \obsref{observation:reductionEffect} shows that the running time is polynomial for fixed~$q$, and that the output instance has the appropriate size. Note that we hide the constant factor~$q$ in the asymptotic notation. It remains to prove that the two instances are equivalent.

If~$S_1 \cup \ldots \cup S_q$ is a partition of~$V(G)$ such that~$G[S_i]$ contains no induced subgraph in~$\Pi$ for all~$i \in [q]$, then that partition can be safely restricted to the vertex set of~$G'$ to yield a solution to the output instance: since~$G'[S_i \cap V(G')]$ is an induced subgraph of~$G[S_i]$, the~$\Pi$-freeness of the latter implies that no set of the restricted partition induces a graph in~$\Pi$. Hence if the input is a \yes-instance, then the output instance is as well. The reverse direction is given by \lemmaref{lemma:reductionPreservesPartition}, which concludes the proof.
\myqed
\end{proof}

\begin{table}[t]
	\centering
{
\small
\begin{tabular}{@{}lll@{}}
\toprule
\textsc{Partition into~$q$} & Forbidden property~$\Pi$ & $c_\Pi$ \\ \midrule
\textsc{Independent Sets} & $\{K_2\}$ & $1$ \\
\textsc{Bipartite Graphs} & Graphs with an odd cycle & $2$ \\
\textsc{Chordal Graphs} & Graphs with a chordless cycle & $3$ \\
\textsc{\F-Minor-Free Graphs} & Graphs with an~$H \in \F$-minor & $\max_{H \in \F} \Delta(H)$ \\
\textsc{Planar Graphs} & Graphs with a~$K_5$ or~$K_{3,3}$ minor & $4$ \\
\textsc{Forests} & Graphs with a cycle & $2$ \\ \bottomrule
\end{tabular}
	\caption{Problems that admit polynomial kernels when parameterized by the size of a given vertex cover, by applying \thmref{theorem:partitioningVC}.}
\label{table:partitionPiTable}
}
\end{table}

The theorem has consequences for a multitude of graph partitioning problems; a sample is presented in \tableref{table:partitionPiTable}. Observe that countless other problems such as \PartitionIntoDistanceHereditary can be captured by the theorem, by using \proposref{proposition:closure} to find new graph properties characterized by few adjacencies.

\begin{corollary} \label{cor:partitionImplications}
All problems in \tableref{table:partitionPiTable} fit into the framework of \thmref{theorem:partitioningVC} and admit polynomial kernels parameterized by the size of a given vertex cover.
\end{corollary}
\begin{proof}
Since the graph properties~$\Pi$ needed to establish the claims in the table were also used in \corollaryref{cor:deletionToPiImplications}, and the preconditions for \thmref{theorem:deletionToPiVC} are stronger than the preconditions to the current theorem, the proofs given there also apply to this case. The table already lists the relevant choice of~$\Pi$ and the resulting~$c_{\Pi}$ needed to apply \thmref{theorem:partitioningVC}. For completeness we state the corresponding choice of polynomial~$p(n)$, and the resulting size bounds.

\textsc{Partition into~$q$ Independent Sets}. Since the forbidden family~$\Pi$ is finite and contains only a single graph on two vertices,~$c_\Pi = 1$ and~$p(n) = 2$ suffices. We obtain a kernel with~$\Oh(|X|^q)$ vertices, which may also be seen as a kernel for \qcoloring parameterized by vertex cover.


\textsc{Partition into~$q$ Bipartite Graphs.} The graphs with an odd cycle are characterized by~$c_\Pi = 2$ adjacencies. The number of vertices in vertex-minimal graphs in this family is at most twice the vertex cover number, so~$p(n) = 2n$ suffices. The resulting kernel size is~$\Oh(|X|^{2q + 1})$ vertices.

\textsc{Partition into~$q$ Chordal Graphs.} The forbidden family is characterized by~$c_\Pi = 3$ adjacencies and the polynomial~$p(n) = 2n$ suffices, resulting in a kernel with~$\Oh(|X|^{3q+1})$ vertices.

\textsc{Partition into~$q$ \F-Minor-Free Graphs.} As shown in the proof of \corollaryref{cor:deletionToPiImplications} the forbidden family is characterized by~$c_\Pi = \max_{H \in \F} \Delta(H)$ adjacencies and the polynomial can be taken to be a linear function whose coefficient depends on~$\F$. We obtain a kernel with~$\Oh(|X|^{q \cdot \Delta+1})$ vertices, where~$\Delta = \max_{H \in \F} \Delta(H)$.

As the last two problems are special cases of the previous item (with~$\F = \{K_5, K_{3,3}\}$ resp.~$\F = \{K_3\}$), this directly shows that we obtain a kernel with $\Oh(|X|^{4q+1})$ and~$\Oh(|X|^{2q+1})$ vertices for the planar and forest partitioning problems, respectively.
\end{proof}

As mentioned in the introduction, \thmref{theorem:partitioningVC} can be considered a strong generalization of the kernel with~$\Oh(|X|^q)$ vertices for \qcoloring parameterized by vertex cover~\cite[Corollary 1]{JansenK11b}. Despite the generality of \thmref{theorem:partitioningVC}, the size of the \qcoloring kernel obtained through \thmref{theorem:partitioningVC} matches that of the \qcoloring kernel given earlier up to constant factors. In the same paper~\cite{JansenK11b} it is proven that for any~$q \geq 4$ and~$\varepsilon > 0$, \qcoloring parameterized by vertex cover does not have kernels of bitsize~$\Oh(|X|^{q - 1 - \varepsilon})$ unless \containment. This shows that in the kernel size bound of \thmref{theorem:partitioningVC}, the appearance of~$q$ in the exponent is unavoidable.

Other partitioning problems that were listed by Garey and Johnson include \PartitionIntoDominatingSets~\cite[GT3]{GareyJ79} (also known as \DomaticNumber), \PartitionIntoHamiltonianSubgraphs~\cite[GT13]{GareyJ79}, and \PartitionIntoPerfectMatchings~\cite[GT16]{GareyJ79}. These problems cannot be expressed in our framework. The last two have trivial polynomial-size kernels parameterized by vertex cover, as one may easily verify that the size of all \yes-instances is bounded polynomially in their vertex cover number. A polynomial kernel can therefore be obtained by simply rejecting instances that are too large. The problem \PartitionIntoDominatingSets may be interesting for further study.

\section{Subgraph Testing versus Minor Testing} \label{section:orderTesting}
Several important graph problems such as \Clique, \LongPath, and \InducedPath, can be stated in terms of testing for the existence of a certain graph~$H$ as an induced subgraph, or as a minor. Note that for these problems, the size of the graph whose containment in~$G$ is tested is part of the input: the problem is polynomial-time solvable for each constant size. We compared the kernelization complexity of induced subgraph- versus minor testing for various types of graphs, parameterized by vertex cover, and found the surprising outcome that the kernelization complexity is often opposite: one variant admits a polynomial kernel while the other does not, assuming \ncontainment. In Sections~\ref{subsection:test:clique}--\ref{subsection:test:matching} we discuss our findings separately for each type of graph whose containment is tested. A summary of our results is given in \tableref{table:orderTesting}. Considering the list of positive and negative results in the table, one might conjecture that testing for an induced $H$-subgraph with~$\vc(H) \in \Oh(1)$ admits a polynomial kernel. In Section~\ref{subsection:test:subgraph:constant:vc} we prove that this implies \containment, and is therefore unlikely. Similarly, the results in the table might lead one to conjecture that testing for any $H$-minor with~$|V(H)| \in \Oh(\vc(G))$ admits a polynomial kernel. However, we prove in Section~\ref{subsection:test:minor:small:vc} that this also implies \containment.

\begin{table}[t]
	\centering
		\begin{tabular}{@{}ll@{}lrl@{}lr@{}}
		\toprule
Graph~$H$ & & \multicolumn{2}{l}{Testing for induced~$H$} & & \multicolumn{2}{l}{Testing for $H$-minor} \\ \midrule
$K_t$ & $\neg$ & $\exists |X|^{\Oh(1)}$ kernel & \cite{BodlaenderJK11} & & $\exists |X|^{\Oh(1)}$ kernel & (Thm.~\ref{theorem:CliqueMinorKernel}) \\
$K_{1,t}$ & & $\exists |X|^{\Oh(1)}$ kernel & (Thm.~\ref{theorem:InducedStarKernel}) & $\neg$ & $\exists |X|^{\Oh(1)}$ kernel & \cite{DomLS09} \\
$K_{s,t}$ & $\neg$ & $\exists |X|^{\Oh(1)}$ kernel & (Thm.~\ref{theorem:InducedBiCliqueNoPoly}) & $\neg$ & $\exists |X|^{\Oh(1)}$ kernel & \cite{DomLS09} \\
$P_t$ & $\neg$ & $\exists |X|^{\Oh(1)}$ kernel & (Thm.~\ref{theorem:inducedPathByVCNoPoly}) & & $\exists |X|^{\Oh(1)}$ kernel & (Thm.~\ref{theorem:largestInducedPiVC}) \\ 
$t \cdot K_2$ & $\neg$ & $\exists |X|^{\Oh(1)}$ kernel & (Thm.~\ref{theorem:inducedMatchingByVCNoPoly}) & & \multicolumn{2}{l}{P-time solvable}\\
	\bottomrule
	\end{tabular}
	\caption{Kernelization complexity of testing for induced~$H$ subgraphs versus testing for~$H$ as a minor, when the graph~$H$ is given as part of the input by specifying the index~$t$. The problems are parameterized by the size of a given vertex cover. Kernel lower bounds are under the assumption that \ncontainment.}
\label{table:orderTesting}
\end{table}

\subsection{Testing for Cliques} \label{subsection:test:clique}
The \Clique problem (i.e., testing for~$K_t$ as an induced subgraph) was one of the first problems known not to admit a polynomial kernel parameterized by the size of a given vertex cover~\cite[Theorem 11]{BodlaenderJK11}. Our main result of this section is a polynomial kernel for the related minor testing problem.

\parproblemdef
{\CliqueMinorTestVC}
{A graph~$G$ with a vertex cover~$X$, and an integer~$t \geq 1$.}
{The size~$|X|$ of the vertex cover.}
{Does~$G$ contain~$K_t$ as a minor?}

\noindent Our polynomial kernel uses reduction rules based on simplicial vertices, inspired by the recent work on kernels for \Treewidth~\cite{BodlaenderJK11b}.

\begin{theorem} \label{theorem:CliqueMinorKernel}
\CliqueMinorTestVC admits a kernel with $\Oh(|X|^4)$ vertices.
\end{theorem}

The remainder of this section is devoted to the proof of the theorem. Firstly, observe that if a graph has a clique~$K_t$ as a minor, then its vertex cover number is at least~$t-1$: taking a minor does not increase the vertex cover number, and~$\vc(K_t) = t-1$. Therefore, we assume that $t\leq |X|+1$, as otherwise we may output a trivial \no-instance. Our algorithm is based on three reduction rules. In the following, we assume that the reduction rules are exhaustively applied in their given order.

\begin{redrule} \label{rule:CliqueMinorFilling}
If there are distinct vertices~$v,w \in X$ such that $vw\notin E(G)$ and there are more than $(|X|+1)^2$ vertices in $V(G)\setminus X$ adjacent both to~$v$ and~$w$, then add the edge $vw$. Output the resulting instance $(G',X,t)$.
\end{redrule}

\begin{lemma}
\ruleref{rule:CliqueMinorFilling} is safe.
\end{lemma}
\begin{proof}
Let~$G'$ be obtained from~$G$ by applying the reduction rule to~$v$ and~$w$. As~$G$ is a subgraph of~$G'$, any clique minor in~$G$ is also contained in~$G'$. Therefore we need to argue that if~$G'$ admits a~$K_t$ minor, then~$G$ admits one as well.

Assume that~$G'$ has a~$K_t$ minor, and let~$G^*$ be a subgraph of~$G'$ containing a~$K_t$ minor model~$\phi$ such that~$|V(G^*)| \leq |V(K_t)| + \vc(G') \cdot (\Delta(K_t) + 1) = t + \vc(G') \cdot t$, whose existence is guaranteed by \proposref{proposition:boundedDegreeMinorModel}. As~$\vc(G') \leq |X|$ it follows that $|\bigcup _{v \in K_t} \phi(v)| \leq t + |X| \cdot t$. Since~$t \leq |X| + 1$ the number of vertices involved in the minor model is at most~$(|X|+1)^2$. Hence by the precondition to the reduction rule, there is a vertex~$y$ adjacent to both~$v$ and~$w$ which is not used in the minor model.

Observe that if~$\phi$ avoids one of~$v$ and~$w$, it is also a clique model in~$G$. Assume then that $v\in \phi(u_1)$ and $w\in \phi(u_2)$; it may happen that~$u_1=u_2$. Now we can transform~$\phi$ into a clique minor model~$\phi'$ in~$G$, by adding~$y$ to $\phi(u_1)$: contraction of the edge~$vy$ in this branch set creates the edge~$vw$ that was missing in~$G$.
\myqed
\end{proof}

Note that exhaustive application of this rule already bounds the number of vertices in $V(G)\setminus X$ that are not simplicial. The next two rules take care of the simplicial vertices.

\begin{redrule} \label{rule:CliqueMinorLargeDegree}
If there exists a simplicial vertex~$s\in V(G)\setminus X$ such that $\deg(s)\geq t-1$, output a trivial \yes-instance.
\end{redrule}

Correctness of \ruleref{rule:CliqueMinorLargeDegree} is obvious, as~$s$ together with its neighborhood already forms a~$K_t$. The following rule is more involved.

\begin{redrule} \label{rule:CliqueMinorSmallDegree}
If there exists a simplicial vertex~$s\in V(G)\setminus X$ such that $\deg(s)<t-1$, delete it. Output the resulting instance~$(G',X,t)$.
\end{redrule}

\begin{lemma} \label{lemma:cliqueMinorSmallDegreeSafe}
\ruleref{rule:CliqueMinorSmallDegree} is safe.
\end{lemma}
\begin{proof}
As~$G'$ is a subgraph of~$G$, any clique minor in~$G'$ also exists in~$G$. Therefore, we need to argue that if $G$ admits a $K_t$ minor, then~$G'$ does as well.

Let~$\phi$ be a clique minor model in~$G$. If~$s$ does not belong to any branch set~$\phi(v)$ for~$v \in K_t$, then~$\phi$ is also a clique minor in~$G'$ and we are done. Assume then that~$s\in \phi(v)$. Observe that~$\phi(v)$ has to contain at least one vertex from~$N_G(s)$, as otherwise we would have that $\phi(v) = \{s\}$ and this~$\phi(v)$ would be able to touch at most~$t-2$ other branch sets. Obtain~$\phi'$ from~$\phi$ by removing~$s$ from~$\phi(v)$ and observe that~$\phi'$ is a~$K_t$ model in~$G'$: all the connections that were introduced by~$s$ are already present in the clique~$N_G(s)$.
\myqed
\end{proof}

The running time of the kernelization algorithm is polynomial, as the presented reduction rules can only add edges inside~$X$ and remove vertices from~$V(G)\setminus X$. Exhaustive application of the reduction rules results in an instance with at most~$(|X|+1)^4$ vertices.

\begin{lemma} \label{lemma:cliqueMinorKernelSizeBound}
If Reduction Rules~\ref{rule:CliqueMinorFilling}--\ref{rule:CliqueMinorSmallDegree} are not applicable, then $|V(G)|\leq (|X|+1)^4$.
\end{lemma}
\begin{proof}
After exhausting Reduction Rules \ref{rule:CliqueMinorLargeDegree} and \ref{rule:CliqueMinorSmallDegree}, there are no simplicial vertices in $V(G)\setminus X$. As \ruleref{rule:CliqueMinorFilling} is not applicable, for each of the at most $\binom{|X|}{2}$ non-edges in $X$ there are at most $(|X|+1)^2$ vertices of $V(G)\setminus X$ adjacent to both endpoints. As every vertex of~$V(G)\setminus X$ is adjacent to the endpoints of some non-edge, $|V(G')| \leq |X| + \binom{|X|}{2} \cdot (|X|+1)^2 \leq (|X|+1)^4$.
\myqed
\end{proof}

This concludes the proof of \thmref{theorem:CliqueMinorKernel}. Let us briefly consider the possibility of extending this result to other graph classes than cliques. \ruleref{rule:CliqueMinorFilling} can be generalized to the setting of testing for any graph of bounded independence number as a minor; cliques are the special case of independence number one. If the graph to be tested has independence number at most~$\alpha$, then we may add an edge between distinct nonadjacent vertices~$v,w$ in~$X$ if there are more than~$(|X|+\alpha)^2$ vertices in~$V(G) \setminus X$ that are adjacent to both~$v$ and~$w$. This rule allows the number of nonsimplicial vertices in the graph to be bounded by a polynomial in the vertex cover size. \ruleref{rule:CliqueMinorLargeDegree} also goes through in the general case; if~$G$ has a simplicial vertex of degree at least~$t-1$, then it has a $t$-clique, and therefore contains all graphs on at most~$t$ vertices as a minor. There seems to be no counterpart of \ruleref{rule:CliqueMinorSmallDegree} in the general case, though. The proof of \thmref{theorem:hminortestbyvch:nopoly} shows that the low-degree simplicial vertices are the hardest to get rid of, since no other types of vertices are needed in that kernelization lower bound construction.

\subsection{Testing for Bicliques}
We now consider the problem of testing for a biclique as an induced subgraph or as a minor. Observe first that if~$G$ is a connected graph on at least three vertices, then the following conditions are equivalent: graph~$G$ has (a) a spanning tree with~$t$ or more leaves, (b) a $K_{1,t}$ minor, (c) a connected dominating set of size at most~$|V(G)| - t$. Hence there is a trivial polynomial-parameter transformation~\cite{Bodlaender09} from \ConnectedDominatingSetVC to \StarMinorTestVC. Dom et al.~\cite[Theorem 5]{DomLS09} showed\footnote{The lower bound they give is for \DominatingSet parameterized by vertex cover, but a trivial transformation extends it to \ConnectedDominatingSet.} that the former problem does not admit polynomial kernels unless \containment. Using the fact that the classical versions of both problems are NP-complete, and the propagation of kernelization lower bounds by polynomial-parameter transformations~\cite[Theorem 8]{BodlaenderTY11}, this implies that \StarMinorTestVC does not admit a polynomial kernel unless \containment.

The situation is more diverse when testing for a biclique as an induced subgraph. If we fix a constant~$c$ and wish to test for a biclique~$K_{c,t}$ as induced subgraph, where~$t$ is part of the input, then this problem admits a polynomial kernel parameterized by vertex cover. The kernel is developed in \sectref{section:inducedstarkernel}. Our main insight is a polynomial-size compression which is obtained by guessing the model of the constant-size partite set within the vertex cover, reducing the problem to the OR of~$\binom{|X|}{c}$ instances of \IndependentSet parameterized by vertex cover. As \IndependentSet parameterized by vertex cover is equivalent to \VertexCover parameterized by the size of a given (suboptimal) vertex cover, each of these can be compressed to a size polynomial in~$|X|$ using \thmref{theorem:deletionToPiVC}. The NP-completeness transformation then results in an instance of the original problem of size~$\Oh(|X|^{\Oh(1)})$ which forms the kernel.

If the sizes of both partite sets are part of the input, then we can no longer obtain a polynomial kernel. In \sectref{section:inducedbycliquelowerbound} we give a cross-composition from \BipartiteBiclique to show that testing for an induced~$K_{s,t}$ subgraph, parameterized by vertex cover, does not admit a polynomial kernel unless \containment.

\subsubsection{Polynomial Kernel for Induced \texorpdfstring{$K_{c,t}$}{K(c,t)}-testing} \label{section:inducedstarkernel}
We give a polynomial kernel for the following problem.

\parproblemdef
{\ConstantBicliqueTest}
{A graph~$G$ with a vertex cover~$X$, and an integer~$t \geq 1$.}
{The size~$|X|$ of the vertex cover.}
{Does~$G$ contain~$K_{c,t}$ as an induced subgraph?}

\noindent Observe that~$c$ is treated as a constant, rather than a variable. The classical version \ConstantBicliqueTestClassical is NP-complete, which will be used in the main proof of this section.


\begin{proposition} \label{proposition:constantBiCliqueTestNPC}
\ConstantBicliqueTestClassical is NP-complete for every constant nonnegative integer~$c$.
\end{proposition}
\begin{proof}
If~$c=0$ then the problem is equivalent to the NP-complete \IndependentSet problem~\cite[GT 20]{GareyJ79}. For~$c \geq 1$ we show how to reduce an instance~$(G,k)$ of \IndependentSet, asking whether~$G$ has an independent set of size at least~$k$, to an equivalent instance of \ConstantBicliqueTestClassical, as follows. Let~$n$ be the number of vertices in~$G$. Form the graph~$G'$ by first adding~$2n+2c$ isolated vertices~$A$ to~$G$, and then adding~$2n + 2c$ independent vertices~$B$ which are adjacent to~$A \cup V(G)$. Then~$G'$ has an induced~$K_{c, k + 2n + 2c}$ subgraph if and only if~$G$ has an independent set of size~$k$. In one direction, it is easy to verify that the vertices of a size-$k$ independent set in~$G$, taken together with~$A \cup B$, induce a~$K_{c, k+2n+2c}$ subgraph in~$G'$. In the other direction, consider a vertex set~$S' \subseteq V(G')$ that induces a~$K_{c, k+2n+2c}$ subgraph. Let~$v \in V(G')$ correspond to a vertex in the size-$c$ side of the biclique, by the isomorphism. Then~$v$ has degree at least~$k + 2n + 2c$ in~$G'$, since that is the degree of vertices in the $c$-side of the biclique. Now observe that for any~$x \in V(G)$, we have~$N_{G'}(x) \subseteq V(G) \cup B$ so~$\deg_{G'}(x) \leq |V(G)| + |B| = n + c$. For~$y \in A$, we have~$N_{G'}(y) \subseteq B$ so~$\deg_{G'}(y) \leq |B| = c$. As~$V(G') = V(G) \cup A \cup B$ this implies that~$v \in B$. As a vertex in the size-$c$ side of~$K_{c, k + 2n + 2c}$ has an independent set of size~$k + 2n + 2c$ in its neighborhood, and~$v$ corresponds to such a vertex by the isomorphism, we find that~$N_{G'}(v)$ contains an independent set of size~$k+2n+2c$. By construction we have~$N_{G'}(v) \subseteq A \cup V(G)$. As~$|A| = 2n + 2c$, there is an independent set of size at least~$k$ in~$N_{G'}(v) \setminus A = V(G)$. Since this set is also independent in~$G$, this proves the equivalence of the two instances and completes the proof.
\myqed
\end{proof}

\noindent With this proposition we can prove the following theorem.

\begin{theorem} \label{theorem:InducedStarKernel}
\ConstantBicliqueTest admits a polynomial kernel for every constant~$c$.
\end{theorem}
\begin{proof}
We may assume that~$t>c$, as otherwise~$K_{c,t}$ is a graph of constant size and we can solve the problem in polynomial time via brute-force. Let~$(G,X,t)$ be the input instance. We provide a polynomial-time algorithm that returns either:
\begin{enumerate}[(i)]
\item one instance of \ConstantBicliqueTest with~$\Oh(|X|^{c+1})$ vertices that is equivalent to~$(G,X,t)$, or \label{case:smallt}
\item at most~$\binom{|X|}{c}$ instances of \IndependentSet, each with~$\Oh(|X|^2)$ vertices, such that~$(G,X,t)$ is a \yes-instance if and only if at least one of them is a \yes-instance.\label{case:larget}
\end{enumerate}
The result of this algorithm gives a polynomial kernel in the following way. In Case~\condref{case:smallt} we can simply output the obtained instance of \ConstantBicliqueTest as the result of the kernelization. For Case~\condref{case:larget} we transform the OR of the \IndependentSet instances into a single instance of \ConstantBicliqueTest of size polynomial in~$|X|$; the result of this transformation is then used as the kernel output. For the transformation we use the intermediate classical problem \OrIndependentSet: ``Given a series of instances of \IndependentSet, is the answer to at least one \yes''? This problem is contained in NP as a nondeterministic Turing machine may simply guess an instance number and a solution, and then verify whether it is correct. We transform the sequence of parameterized \IndependentSet instances into a single instance of \OrIndependentSet of total bitsize polynomial in~$|X|$, by appending all the instances and writing their parameter values in unary. As there are~$\binom{|X|}{c}$ instances, each with $\Oh(|X|^2)$ vertices, this results in a classical instance of \OrIndependentSet of bitsize polynomial in~$|X|$. As \OrIndependentSet is contained in NP and \ConstantBicliqueTestClassical is NP-complete, we may transform this \OrIndependentSet instance in polynomial time to an \ConstantBicliqueTestClassical instance, incurring only a polynomial blowup in instance size. As the \ConstantBicliqueTestClassical instance at this point has size polynomial in~$|X|$, we may simply use the entire graph as the vertex cover~$X'$ to make an instance of \ConstantBicliqueTest, of size and parameter bounded by a polynomial in~$|X|$; this forms the output of the kernelization procedure.

Hence, we are left with presenting the algorithm achieving goal~\condref{case:smallt} or~\condref{case:larget} in polynomial time. In the following, whenever we assume that~$(G,X,t)$ is a \yes-instance, we fix some induced~$K_{c,t}$ subgraph of~$G$ and denote its bipartition by~$(A,B)$, where~$|A|=c$ and~$|B|=t$. First, we exhaustively apply the following reduction rule. For every vertex~$v\in V(G)$ we check whether its neighborhood contains an independent set of size~$c$. If this is not the case, we may safely delete this vertex as it cannot be contained in any induced~$K_{c,t}$; note that in this step we use the assumption that~$t>c$ to verify that the vertex cannot be in part~$A$ of the solution, either. This check can be done in polynomial time by iterating through all the subsets of~$N_G(v)$ of size~$c$. From now on we may assume that each vertex of the graph has an independent set of size~$c$ in its neighborhood.

Observe that if~$|V(G)\setminus X|\geq t\cdot \binom{|X|}{c}$, then~$(G,X,t)$ is a \yes-instance, as some~$t$ vertices of~$V(G)\setminus X$ are adjacent to the same independent set of size~$c$ in~$X$. In this case we output a trivial \yes-instance. Moreover, if this is not the case but~$t\leq |X|$, then~$|V(G)|\leq |X|+|X|\cdot \binom{|X|}{c}=\Oh(|X|^{c+1})$ and we may output the graph obtained so far as the kernel in Case~\condref{case:smallt}.

We are left with the case that~$t>|X|$. Note that if~$(G,X,t)$ is a \yes-instance, then part~$B$ has to contain at least one vertex from~$V(G)\setminus X$, which means that~$A\subseteq X$. For each subset~$A'\subseteq X$ of size~$c$ that induces an independent set, we construct an instance~$(G_{A'},X_{A'},t)$ of \IndependentSetByVC, by taking~$G_{A'}=G[\bigcap_{v\in A'} N_G(v)]$ and~$X_{A'}=X\cap V(G_{A'})$. Observe that if~$(G,X,t)$ has a solution with~$A=A'$, then $(G_{A'},X_{A'},t)$ is a \yes-instance as the corresponding part~$B$ is contained in~$\bigcap_{v\in A'} N_G(v)$. On the other hand, if~$G_{A'}$ contains an independent set~$B'$ of size~$t$, then~$A'\cup B'$ induces a~$K_{c,t}$ in~$G$. Therefore,~$(G,X,t)$ is a \yes-instance if and only if then at least one of the instances~$(G_{A'},X_{A'},t)$ is a \yes-instance. Observing that \IndependentSetByVC is equivalent to \VertexCoverByVC (by going to the dual target value~$k' := n - k$, while keeping the parameter~$|X|$ the same) we can apply the kernelization algorithm for \VertexCoverByVC from \thmref{theorem:deletionToPiVC} to every instance~$(G_{A'},X_{A'},t)$. Transforming the result back into \IndependentSet instances, we thus obtain a sequence of instances of \IndependentSet with~$\Oh(|X|^2)$ vertices each, that can be returned in Case~\condref{case:larget}.
\myqed
\end{proof}

The guessing steps used in the kernelization above are reminiscent of a Turing kernel. We are effectively creating a compression (in the language of Harnik and Naor~\cite{HarnikN10}) for \ConstantBicliqueTest by reducing it to the OR of a sequence of~$\poly(|X|)$ \IndependentSet instances of size~$\poly(|X|)$. The connection to Turing kernelization is further explored in the conclusion.

\subsubsection{Kernel Lower Bound for Induced \texorpdfstring{$K_{s,t}$}{K(s,t)}-testing} \label{section:inducedbycliquelowerbound}
In this section we prove that the requirement that~$c$ is kept fixed in the definition of \ConstantBicliqueTest is essential for obtaining a polynomial kernel. We consider the variant where the sizes of both partite sets are part of the input, and establish a lower bound. The problem we study is formally defined as follows.

\parproblemdef
{\VariableBicliqueTest}
{A graph~$G$ with vertex cover~$X \subseteq V(G)$ and integers~$s,t \geq 1$.}
{The size~$|X|$ of the vertex cover.}
{Does~$G$ contain~$K_{s,t}$ as an induced subgraph?}

\noindent The crucial difference with \ConstantBicliqueTest is that the value~$s$ is part of the input, rather than a constant. We base our cross-composition on the balanced biclique problem in bipartite graphs.

\problemdef
{\BipartiteBiclique}
{A bipartite graph~$G$ with partite sets~$A \cup B$, and an integer~$k \geq 1$.}
{Are there subsets~$S \subseteq A$ and~$T \subseteq B$ such that~$G[S \cup T]$ is a biclique, and~$|S| = |T| = k$?}

\noindent The problem is known to be NP-complete~\cite[GT24]{GareyJ79} and thus suitable for a cross-composition.

\begin{theorem} \label{theorem:InducedBiCliqueNoPoly}
\VariableBicliqueTest does not admit a polynomial kernel unless \containment.
\end{theorem}
\begin{proof}
We prove that \BipartiteBiclique cross-com\-poses into \VariableBicliqueTest, which suffices to establish the claim by \thmref{crossCompositionNoKernel}.
Define a polynomial equivalence relation \eqvr as follows. Two strings in~$\Sigma^*$ are equivalent if (a) they both encode malformed instances, or (b) they encode valid instances~$(G_1, A_1, B_1, k_1)$ and~$(G_2, A_2, B_2, k_2)$ of \BipartiteBiclique such that $|A_1|=|A_2|$, $|B_1|=|B_2|$ and $k_1=k_2$. This relation \eqvr partitions a set of instances on at most~$n$ vertices each into~$\Oh(n^3)$ equivalence classes, and is therefore a polynomial equivalence relation.

We compose instances which are equivalent under \eqvr. So the input consists of~$r$ instances~$(G_1, A_1, B_1, k), \ldots, (G_r, A_r, B_r, k)$ of \BipartiteBiclique which all agree on the number of vertices in each partite set, and on the value of~$k$. By duplicating some instances we may assume without loss of generality that~$r$ is a power of two. Let~$n := |B_1| = \ldots = |B_r|$ and~$m := |A_1| = \ldots = |A_r|$. For~$i \in [r]$ label the vertices in~$B_i$ as~$b_{i,1}, \ldots, b_{i,n}$. We build a graph~$G^*$ with vertex cover~$X^*$ as follows.
\begin{itemize}
\item Initialize~$G^*$ as the disjoint union of the input graphs~$G_1, \ldots, G_r$.
\item For each~$j \in [n]$, identify the vertices~$b_{1,j}, \ldots, b_{r,j}$ into a single vertex~$b^*_j$. Let~$B^* := \{b^*_1, \ldots, b^*_n\}$ contain the resulting vertices, and observe that at this stage in the construction~$G^*[A_i \cup B^*]$ is isomorphic to~$G_i$ for~$i \in [r]$.
\item For~$j \in [\log r]$, add to~$G^*$ a biclique~$C_j$ isomorphic to~$K_{n+1,n+1}$, with partite sets denoted by~$P_j$ and~$Q_j$,~$|P_j|=|Q_j|=n+1$. The set of vertices corresponding to one value of~$j$ will be called the \emph{bit selector} of~$j$ as it will be used in valid solutions to select the bitvalue of the binary representation of the input instance corresponding to this solution.
\item For~$j \in [\log r]$, make the vertices~of~$P_j$ adjacent to the vertices of~$A_i$ if the $j$-th bit in the binary representation of number~$i$ is a one. Similarly, make the vertices~$Q_j$ adjacent to~$A_i$ if the $j$-th bit of~$i$ is a zero.
\item Add a set~$D$ of~$(n+1)(1 + 2 \log r)$ vertices, adjacent to all the vertices of~$B^*$ and all the vertices of all the bit selectors.
\item Let~$X^*$ contain the vertices of~$B^*$ and all the vertices of all the bit selectors. Observe that~$|X^*| = n + 2(n+1) \log r = |D| - 1$ and that~$G^* - X^*$ is an independent set containing all the sets~$A_i$ and the set~$D$; hence~$X^*$ is a vertex cover of~$G^*$ whose size is suitably bounded for a cross-composition.
\end{itemize}
The construction is completed by setting~$s:=k+(n+1)\log r$ and~$t:=k+(n+1)(1 + 2 \log r) = k + |D|$. We now prove the completeness and soundness of the composition via two claims.

\begin{claim}
If for some~$i \in [r]$ the instance~$(G_i, A_i, B_i, k)$ is a \yes-instance of \BipartiteBiclique, then~$(G^*, X^*, s, t)$ is a \yes-instance of \VariableBicliqueTest.
\end{claim}
\begin{claimproof}
Let~$S\subseteq A_i$ and~$T\subseteq B_i$ be such that~$|S|=|T|=k$ and~$G_i[S\cup T]$ is a biclique. Let~$T'$ be the image of~$T$ in the identifications, i.e.,~$T'=\{b^*_j \mid b_{i,j}\in T\}$. For~$j\in [\log r]$ define~$R_j:=P_j$ if the~$j$-th bit of binary encoding of~$i$ is equal to one, and define~$R_j:=Q_j$ otherwise. We claim that the set~$S \cup T'\cup D\cup \bigcup_{j\in [\log r]} R_j$ induces a biclique in~$G^*$, with~$T'\cup \bigcup_{j\in [\log r]} R_j$ as one partite set and~$S\cup D$ as the second. Indeed, observe that:
\begin{itemize}
\item $D,S,T',\bigcup_{j\in [\log r]} R_j$ are independent sets by the construction of~$G^*$;
\item there is no edge between~$D$ and~$S$;
\item there is no edge between~$\bigcup_{j\in [\log r]} R_j$ and~$T'$;
\item $D$ is adjacent to the whole set~$X^*$, so in particular to~$T'\cup \bigcup_{j\in [\log r]} R_j$;
\item as~$S\subseteq A_i$, by the construction of~$G^*$ we have that every vertex of~$S$ is adjacent to every vertex of~$R_j$, for all~$j\in [\log r]$;
\item all vertices in~$S$ are adjacent to all vertices of~$T'$, as~$G_i[S \cup T]$ is a biclique and~$G^*[A_i \cup B^*]$ is isomorphic to~$G_i$.
\end{itemize}
We conclude the proof by checking that~$|T'\cup \bigcup_{j\in [\log r]} R_j|=k+(n+1)\log r = s$ and~$|S\cup D|=k+|D|=t$.
\end{claimproof}

\begin{claim}
If~$(G^*, X^*, s, t)$ is a \yes-instance of \VariableBicliqueTest, then for some~$i \in [r]$ the instance~$(G_i, A_i, B_i, k)$ is a \yes-instance of \BipartiteBiclique.
\end{claim}
\begin{claimproof}
Assume that there exist sets~$S^*$ and~$T^*$,~$|S^*|=s$ and~$|T^*|=t$, such that~$G^*[S^*\cup T^*]$ is a biclique with~$S^*$ and~$T^*$ as partite sets.
As~$|T^*|=t\geq |D|>|X^*|$, the set~$T^*\setminus X^*$ is nonempty. This means that in~$T^*$ there is a vertex with the whole neighborhood entirely contained in~$X^*$, so~$S^*\subseteq X^*$. From every pair~$(P_j,Q_j)$, for~$j\in [\log r]$, the independent set~$S^*$ can have a nonempty intersection with at most one of them. Assume that for some~$j$ we have~$P_j\cap S^*= Q_j\cap S^*=\emptyset$. It follows that~$|S^*|\leq (n+1)\log r - (n+1) + |B^*|<(n+1)\log r +k = s$, which is a contradiction. Hence, for all~$j\in [\log r]$ the set~$S^*$ has a nonempty intersection with exactly one set of~$P_j$ and~$Q_j$. Moreover, observe that~$|S^*\cap \bigcup_{j\in [\log r]} (P_j\cup Q_j)|\leq (n+1)\log r$, so~$|S^*\cap B^*|\geq k$.

Define~$i$ as an integer with~$\log r$ binary digits, such that the~$j$-th bit is equal to one if~$P_j\cap S^*\neq \emptyset$ and is equal to zero if~$Q_j\cap S^*\neq \emptyset$. By the construction of~$G^*$, the set~$A_i$ is the only set from~$\{A_1,\ldots,A_r\}$ which contains vertices simultaneously adjacent to all vertices from~$S^*$ contained in bit selectors. There is no edge between~$B^*$ and bit selectors, so we infer that~$T^*\subseteq A_i\cup D$. As~$|T^*|=k+|D|$, we infer that~$|T^*\cap A_i|\geq k$.

Recall that~$G^*[B^*\cup A_i]$ is isomorphic to~$G_i$, hence~$(T^*\cap A_i) \cup (S^*\cap B^*)$ induces a biclique in a graph isomorphic to~$G_i$. As~$|T^*\cap A_i|,|S^*\cap B^*|\geq k$, we infer that~$G_i$ is a \yes-instance of \BipartiteBiclique.
\myqed
\end{claimproof}
As this proves that the output instance acts as the logical OR of the inputs, it concludes the cross-composition and proves a kernel lower bound by \thmref{crossCompositionNoKernel}.
\myqed
\end{proof}

\subsection{Testing for Paths}
We turn our attention to testing for the containment of a path. Since a graph contains~$P_t$ as a minor if and only if it contains~$P_t$ as a subgraph, testing for a~$P_t$ minor is equivalent to the \LongPath problem and hence has a polynomial kernel parameterized by vertex cover, through \thmref{theorem:largestInducedPiVC}. The related induced subgraph testing problem, defined formally below, is however unlikely to admit a polynomial kernel. 

\parproblemdef
{\InducedPathByVC}
{A graph~$G$ with a vertex cover~$X$, and an integer~$k \geq 1$.}
{The size~$|X|$ of the vertex cover.}
{Is there a set~$S \subseteq V(G)$ of size at least~$k$ such that~$G[S]$ is a simple path?}
Using cross-composition, we start from the following classical problem.
\problemdef
{\HamSTpath}
{A graph~$G$ with distinct vertices~$s$ and~$t$.}
{Is there a Hamiltonian path from~$s$ to~$t$ in~$G$?}

Before we proceed to the formal description, let us shed some light on the intuition behind the proof. We cross-compose~$r$ instances of \HamSTpath into a single instance of \InducedPathByVC. The main idea behind the construction is to create an instance containing three paths~$P_A,P_B,P_C$ of consecutive degree-two vertices, such that any sufficiently long induced path traverses all these paths. The only connections between~$P_A$ and~$P_B$ can be made by visiting a vertex~$z_i$ outside the vertex cover; there is one such vertex~$z_i$ for each input instance. Hence, the connection between~$P_A$ and~$P_B$ selects an instance. The connection between~$P_B$ and~$P_C$ serves for checking that the selected instance can indeed be solved. We create a universal gadget in which the connection between~$P_B$ and~$P_C$ has to be realized. Using the inducedness requirement, we encode adjacency matrices of the input instances into the adjacencies between vertices~$z_i$ and the universal gadget: selection of some~$z_i$ ``carves out'' the $i$-th instance from the universal gadget by forbidding usage of vertices adjacent to~$z_i$. We now proceed to the formal description of the composition.

\begin{theorem} \label{theorem:inducedPathByVCNoPoly}
\InducedPathByVC does not admit a polynomial kernel unless \containment.
\end{theorem}
\begin{proof}
By \thmref{crossCompositionNoKernel} and the NP-completeness of \HamSTpath~\cite[GT 39]{GareyJ79}, it is sufficient to show that \HamSTpath cross-composes into \InducedPathByVC. We define a polynomial equivalence relation \eqvr as follows. We say that two strings in~$\Sigma^*$ are equivalent if (a) they both encode malformed instances, or (b) they encode valid instances~$(G_1, s_1, t_1)$ and~$(G_2, s_2, t_2)$ of \HamSTpath such that~$|V(G_1)| = |V(G_2)|$. This implies that \eqvr partitions a set of instances on at most~$n$ vertices each into~$\Oh(n)$ equivalence classes, and is therefore a polynomial equivalence relation.

We show how to compose a set of instances which are equivalent under \eqvr. So the input consists of~$r$ instances~$(G_1, s_1, t_1), \ldots, (G_r, s_r, t_r)$ of \HamSTpath such that~$|V(G_i)| = n$ for~$i \in [r]$. We may assume that~$n \geq 9$, since we can solve smaller instances in constant time, reducing to a constant-size \yes- or \no-instance. For~$i \in [r]$ label the vertices in~$V(G_i)$ as~$v_1, \ldots, v_n$ such that~$s_i = v_1$ and~$t_i = v_n$. We build a graph~$G^*$ with vertex cover~$X^*$ as follows.
\begin{enumerate}
	\item Add three simple paths~$P_A, P_B$ and~$P_C$ to~$G^*$, containing~$n^3$ vertices each. Let the endpoints of these paths be~$x_A, y_A, x_B, y_B$ and~$x_C, y_C$ respectively.
	\item For~$j \in [n]$ add a vertex~$v^*_j$ to~$G^*$.
	\item For~$\{j,h\} \in \binom{[n]}{2}$ add a vertex~$e_{j,h}$ to~$G^*$ and make it adjacent to~$v^*_j$ and~$v^*_h$.
	\item For~$i \in [r]$, do the following. Add a vertex~$z_i$ to~$G^*$. For all pairs~$\{j,h\} \in \binom{[n]}{2}$ such that~$v_jv_h \not \in E(G_i)$ add the edge~$z_ie_{j, h}$ to~$G^*$.
	\item Make~$y_A$ and~$y_B$ adjacent to all vertices~$z_i$ for~$i \in [r]$. 
	\item Make~$x_B$ adjacent to~$v^*_1$, and make~$x_C$ adjacent to~$v^*_n$. This concludes the construction of~$G^*$, which is illustrated in \imgref{inducedPathPicture}. 
\end{enumerate}
We define a set~$X^* := V(G^*) \setminus \{ z_i \mid i \in [r] \}$. Since we did not add any edges between the~$z$-vertices, they form an independent set and therefore~$X^*$ is a vertex cover of~$G^*$. It is easy to verify that the size of~$X^*$ is polynomial in~$n$, and therefore the size of the parameter~$|X^*|$ is suitably bounded for a cross-composition. We set~$k^* := 3 n^3 + 2n$. The construction can be performed in polynomial time, so it remains to prove that~$(G^*, X^*, k^*)$ is \yes if and only if one of the input instances is \yes. We first establish some properties of the constructed instance.

\inducedPathPicture

\begin{claim}
Let~$S^* \subseteq V(G)$ induce a simple path in~$G^*$, and let~$P^* := G^*[S^*]$.
\begin{enumerate}
	\item $|S^* \setminus (V(P_A) \cup V(P_B) \cup V(P_C))| < n^3 - 3$. \label{solutionOutsidePathsSmall}
	\item If there is a path~$P_W \in \{P_A, P_B, P_C\}$ such that~$|S^* \cap V(P_W)| \leq 3$ then~$|S^*| < k^*$. \label{notAllPathsUsedSmall}
\end{enumerate}
\end{claim}
\begin{claimproof}
Define $\hat{G^*} := G^* - (V(P_A) \cup V(P_B) \cup V(P_C))$. 

\condref{solutionOutsidePathsSmall} For each of the paths~$P_A, P_B, P_C$ there are at most two vertices on the path which have neighbors outside the path. Hence if we take the path~$P^*$, then deleting the vertices of~$V(P_A)$ from~$P^*$ splits the path into at most three pieces, increasing the number of connected components by at most two. This also holds for~$P_B$ and~$P_C$. Hence~$\hat{P^*} := P^* - (V(P_A) \cup V(P_B) \cup V(P_C))$ is an induced linear forest in~$\hat{G^*}$ containing no more than seven connected components, with~$|S^*| = |V(\hat{P}^*)|$. Each connected component of~$\hat{P^*}$ is an induced path in~$\hat{G^*}$. Since the set~$\hat{X^*} := \{ v^*_i \mid i \in [n] \} \cup \{ e_{j,h} \mid \{j,h\} \in \binom{[n]}{2} \}$ is a vertex cover for~$\hat{G^*}$ of size~$n + \binom{n}{2}$ it follows by Proposition~\ref{proposition:vcpathcycle} that each connected component of~$\hat{P^*}$ has at most~$2n + 2\binom{n}{2} + 1$ vertices. Since the number of connected components is at most seven, the number of vertices in~$\hat{P^*}$ is at most~$7 \cdot (2n + 2\binom{n}{2} + 1)$, which is less than~$n^3 - 3$ for~$n \geq 9$.

\condref{notAllPathsUsedSmall} Assume that~$S^*$ contains at most three vertices from~$P_A$; the other two cases will be completely analogous.
\begin{align*}
|S^*| =& |S^* \setminus (V(P_A) \cup V(P_B) \cup V(P_C))| + \\ & |S^* \cap V(P_A)| + |S^* \cap V(P_B)| + |S^* \cap V(P_C)| \\
<& (n^3 - 3) + |S^* \cap V(P_A)| + \\ & |S^* \cap V(P_B)| + |S^* \cap V(P_C)| & \mbox{By \condref{solutionOutsidePathsSmall}.} \\
\leq& (n^3 - 3) + 3 + |S^* \cap V(P_B)| + |S^* \cap V(P_C)| & \mbox{By assumption.} \\
\leq& (n^3 - 3) + 3 + n^3 + n^3 & \mbox{By definition of~$G^*$.} \\
\leq& k^*. & \mbox{By definition of~$k^*$.}
\end{align*}
Hence if there is one path among~$\{P_A, P_B, P_C\}$ such that~$S^*$ contains at most three vertices on it, then~$|S^*| < k^*$.
\end{claimproof}

We now prove that~$(G^*, X^*, k^*)$ indeed acts as the OR of the input instances. For the first direction, assume that~$G^*$ has a path on at least~$k^*$ vertices induced by the vertex set~$S^*$. Let~$P^* := G^*[S^*]$ be the path induced by~$S^*$. By \condref{notAllPathsUsedSmall} the set~$S^*$ contains at least three vertices on each of the paths~$P_A, P_B, P_C$. Since~$P_A$ and~$P_C$ each contain exactly one vertex which has neighbors outside the path, it is easy to see that~$S^* \cup V(P_A) \cup V(P_C)$ is also an induced path; hence we may assume without loss of generality that~$S^*$ contains all vertices of~$P_A$ and~$P_C$, which means that the endpoints of~$P^*$ must be the vertices~$x_A$ and~$y_C$ since they have degree one in~$G^*$. Since no endpoint of~$P^*$ can lie on~$P_B$, and~$S^*$ contains at least three vertices on~$P_B$, it follows that~$S^*$ must contain all vertices of~$P_B$ since the internal vertices on that path do not have neighbors outside the path. Hence~$V(P_A) \cup V(P_B) \cup V(P_C) \subseteq S^*$. Since the only neighbors of vertex~$y_A$ are the vertices~$z_i$ for~$i \in [r]$ and the single neighbor on the path~$P_A$, the path~$P^*$ must contain an edge~$y_Az_{i^*}$ for some~$i^* \in [r]$ since~$y_A$ must have two neighbors on the path. By construction of~$G^*$ we know that~$\{z_{i^*}, y_B\}$ is an edge in~$G^*$. This implies that if we traverse the path~$P^*$ starting from the endpoint~$x_A$ then we traverse~$P_A$, visit~$z_{i^*}$, and then go to~$y_B$. Since all vertices of~$P_B$ are in~$S^*$ the path then traverses~$P_B$ until it reaches~$x_B$. The unique neighbor~$v^*_1$ of~$x_B$ not on~$P_B$ must be the successor of~$x_B$ on the path~$P^*$. The path now visits some more vertices. Since all vertices of~$P_C$ are contained in~$S^*$, and~$x_C$ is the only vertex of~$P_C$ adjacent to vertices not on~$P_C$, the path~$P^*$ must finish by reaching~$x_C$ and traversing~$P_C$.

Let us now consider the subpath~$P_{i^*}$ of~$P^*$ which starts at the successor of~$x_B$ on the path, and stops with the predecessor of~$x_C$ on the path. The successor of~$x_B$ must be~$v^*_1$, and the predecessor of~$x_C$ must be~$v^*_n$, since vertices~$x_B$ and~$x_C$ have degree two in~$G^*$; hence~$P_{i^*}$ is an induced path from~$v^*_1$ to~$v^*_n$. Since there are~$2 n^3 + 1$ vertices on~$P^*$ before~$x_B$, and~$n^3$ vertices on~$P^*$ on the final part from~$x_C$ to the endpoint, the subpath~$P_{i^*}$ must contain at least~$k^* - 3 n^3 - 1 = 2n - 1$ vertices. Since the vertices~$y_A$ and~$y_B$ are contained in~$S^*$ and are adjacent to~$z_{i^*}$, the set~$S^*$ cannot contain any other vertices adjacent to~$z_{i^*}$ (otherwise these would induce an edge not on the path~$P^*$). This implies that in particular, $S^*$ cannot contain vertices~$e_{j,h}$ for which~$v_jv_h \not \in E(G_{i^*})$ since these were made adjacent to~$z_{i^*}$ in the construction. The set~$S^*$ cannot contain any vertices~$z_i$ for~$i \neq i^*$, since all such vertices are adjacent to~$y_A, y_B \in S^*$ and together with~$z_{i^*} \in S^*$ such a vertex~$z_i$ would induce a cycle. This shows that the subpath~$P_{i^*}$ can contain only vertices~$v^*_j$ for~$j \in [n]$, and vertices~$e_{j,h}$ for~$\{v_j, v_h\} \in E(G_{i^*})$. Since the edge vertices~$e_{j,h}$ are only adjacent to the vertices which form their endpoints, it now follows that the edge set~$\{ v_jv_h \mid e_{j,h} \in S^* \}$ is a path in~$G_{i^*}$ between~$v_1 = s_{i^*}$ and~$v_n = t_{i^*}$ containing~$n-1$ edges and~$n$ vertices, which implies that~$G_{i^*}$ has a Hamiltonian $v_1 - v_n$ path and proves that~$G_{i^*}$ is \yes.

For the reverse direction, assume that the set~$C_{i^*} \subseteq E(G_{i^*})$ are the edges on a Hamiltonian~$v_1 - v_n$ path in~$G_{i^*}$. Then it is straightforward to verify using the construction of~$G^*$ that~$S^* := V(P_A) \cup V(P_B) \cup V(P_C) \cup \{z_{i^*} \} \cup \{ v^*_i \mid i \in [n] \} \cup \{ e_{j,h} \mid v_jv_h \in C_{i^*} \}$ induces a simple path in~$G^*$ and has size~$k^*$. This concludes the proof.
\myqed
\end{proof}

\subsection{Testing for Matchings} \label{subsection:test:matching}
Matchings (i.e., disjoint unions of~$K_2$'s) are the last type of graphs whose containment testing we consider. It is not difficult to see that~$G$ has a~$t \cdot K_2$ minor if and only if~$G$ has a matching of size~$t$, and hence we can solve the minor-testing variant of this containment problem in polynomial time by simply computing a maximum matching. On the other hand, finding an induced matching is a classic NP-complete problem and we give evidence that it does not admit a polynomial kernel parameterized by vertex cover. In the next section we use a bit-selector strategy to cross-compose \BipartiteInducedMatching into our target problem, exploiting the inducedness requirement to allow the bit selector to isolate a solution corresponding to a single input instance.

\subsubsection{Kernelization Lower Bound for Induced Matching}
Recall that an \emph{induced matching} in a graph~$G$ is a matching~$Y \subseteq E(G)$ such that no edge in~$E(G) \setminus Y$ connects the endpoints of two edges of~$Y$, or equivalently, such that all connected components of the subgraph induced by the endpoints of~$Y$ are isomorphic to~$K_2$. The \emph{size} of an induced matching is measured in terms of the number of edges in it. The goal of this section is to prove a superpolynomial kernel lower bound for the following problem.
\parproblemdef
{\InducedMatchingByVC}
{A graph~$G$ with a vertex cover~$X$, and an integer~$k \geq 1$.}
{The size~$|X|$ of the vertex cover.}
{Is there an induced matching~$Y \subseteq E(G)$ in~$G$ of size at least~$k$?}
Using the technique of cross-composition, we start from the following related classical problem.
\problemdef
{\BipartiteInducedMatching}
{A bipartite graph~$G$ with partite sets~$A \cup B$, and an integer~$k \geq 1$.}
{Is there an induced matching~$Y \subseteq E(G)$ in~$G$ of size at least~$k$?}
The cross-composition embeds the OR of bipartite instances into a single instance of the parameterized problem with a small parameter value. The construction is based on a bit masking scheme that represents the indices of the~$r$ input instances by~$\log r$ bits, as in the proof of Theorem~\ref{theorem:InducedBiCliqueNoPoly}. We use repeated structures in the constructed graph to simulate heavy-weight edges.
\begin{theorem} \label{theorem:inducedMatchingByVCNoPoly}
\InducedMatchingByVC does not admit a polynomial kernel unless \containment.
\end{theorem}
\begin{proof}
We prove that \BipartiteInducedMatching cross-composes into \InducedMatchingByVC, which suffices to establish the claim by \thmref{crossCompositionNoKernel} and the NP-completeness of the classical problem~\cite{Cameron89}. Define a polynomial equivalence relation \eqvr as follows. Two strings in~$\Sigma^*$ are equivalent if (a) they both encode malformed instances, or (b) they encode valid instances~$(G_1, A_1, B_1, k_1)$ and~$(G_2, A_2, B_2, k_2)$ of \BipartiteInducedMatching such that~$|A_1| = |A_2|$,~$|B_1| = |B_2|$ and~$k_1 = k_2$. This relation \eqvr partitions a set of instances on at most~$n$ vertices each into~$\Oh(n^3)$ equivalence classes, and is therefore a polynomial equivalence relation.

We compose instances which are equivalent under \eqvr. So the input consists of~$r$ instances~$(G_1, A_1, B_1, k), \ldots, (G_r, A_r, B_r, k)$ of \BipartiteInducedMatching which all agree on the number of vertices in each partite set, and on the value of~$k$. By duplicating some instances we may assume without loss of generality that~$r$ is a power of two. Let~$n := |B_1| = \ldots = |B_r|$. For~$i \in [r]$ label the vertices in~$B_i$ as~$b_{i,1}, \ldots, b_{i,n}$. We build a graph~$G^*$ with vertex cover~$X^*$ as follows.
\begin{itemize}
	\item Initialize~$G^*$ as the disjoint union of the input graphs~$G_1, \ldots, G_r$.
	\item For each~$j \in [n]$, identify the vertices~$b_{1,j}, \ldots, b_{r,j}$ into a single vertex~$b^*_j$. Let~$B^* := \{b^*_1, \ldots, b^*_n\}$ contain the resulting vertices, and observe that at this stage in the construction~$G^*[A_i \cup B^*]$ is isomorphic to~$G_i$ for~$i \in [r]$.
	\item For~$j \in [\log r]$, add vertices~$\{ x_{s, j}, y_{s,j}, z_{s,j} \mid s \in [n] \}$ to~$G^*$ and turn each triplet into a clique. As in the proof of \thmref{theorem:InducedBiCliqueNoPoly} the set of vertices corresponding to one value of~$j$ is the \emph{bit selector} of~$j$.
	\item For~$j \in [\log r]$, make the vertices~$\{x_{s,j} \mid s \in [n] \}$ adjacent to the vertices~$A_i$ if the $j$-th bit in the binary representation of number~$i$ is a one. Similarly, make the vertices~$\{y_{s,j} \mid s \in [n] \}$ adjacent to~$A_i$ if the $j$-th bit of $i$ is a zero. Let~$\{ x_{s,j}z_{s,j} \mid s \in [n] \}$ be the \emph{$x$-edges of position~$j$}, and let~$\{ y_{s,j}z_{s,j} \mid s \in [n] \}$ be the \emph{$y$-edges of position~$j$}.
	\item Let~$X^*$ contain the vertices of~$B^*$ and all the vertices of all the bit selectors. Observe that~$|X^*| = n + 3n \log r$ and that~$G^* - X^*$ is an independent set containing all the sets~$A_i$; hence~$X^*$ is a vertex cover of~$G^*$ whose size is suitably bounded for a cross-composition.
\end{itemize}
The construction is completed by setting~$k^* := k + n \log r$ and using the instance~$(G^*, X^*, k^*)$ as the output of the cross-composition. We will need the following structural claim.

\begin{claim}
$G^*$ has a maximum induced matching~$Y^* \subseteq E(G^*)$ such that for every~$j \in [\log r]$ and~$s \in [n]$, if~$Y^*$ contains an edge incident with the triple~$\{x_{s,j}, y_{s,j}, \linebreak[1] z_{s,j}\}$ then~$Y^*$ contains~$x_{s,j}z_{s,j}$ or~$y_{s,j}z_{s,j}$.
\end{claim}
\begin{claimproof}
Suppose~$Y^*$ is a maximum induced matching containing an edge incident with the triple~$\{x_{s,j}, y_{s,j}, z_{s,j}\}$ for some choice of~$j$ and~$s$, but the edge is neither $x_{s,j}z_{s,j}$ nor~$y_{s,j}z_{s,j}$. As the triple forms a clique in~$G^*$, by the induced property of $Y^*$ it follows that~$Y^*$ contains \emph{at most} one edge incident with it. Using the starting assumption we then find that~$Y^*$ contains \emph{exactly} one edge incident with the triple. Since vertex~$z_{s,j}$ is only adjacent to~$x_{s,j}$ and~$y_{s,j}$ we find that the edge~$e \in Y^*$ incident with the triple, is incident with at least one of the vertices~$x_{s,j}$ or~$y_{s,j}$. If~$e$ is incident with~$x_{s,j}$ then we may replace it by the edge~$x_{s,j}z_{s,j}$ to obtain another induced matching; $z_{s,j}$ was not matched before, and is not adjacent to any matched vertices except~$x_{s,j}$. Similarly we may replace~$e$ by~$\{y_{s,j}, z_{s,j}\}$ if~$e$ is incident with~$y_{s,j}$. As this replacement step can be performed independently for each triple, the claim follows.
\end{claimproof}

\begin{claim}
$G^*$ has a maximum induced matching~$Y^* \subseteq E(G^*)$ such that for every bit position~$j \in [\log r]$, either all the $x$-edges of position~$j$ are in~$Y^*$, or all the $y$-edges of position~$j$ are in~$Y^*$.
\end{claim}
\begin{claimproof}
Consider a maximum induced matching~$Y^*$ in~$G^*$, and assume there is some bit position~$j \in [\log r]$ for which the claim does not hold. By the previous claim we may assume that if~$Y^*$ contains an edge incident with a triple~$\{x_{s,j}, y_{s,j}, \linebreak[1] z_{s,j}\}$, then it is the $x$-edge ~$x_{s,j}z_{s,j}$ or the $y$-edge~$y_{s,j}z_{s,j}$.

If at least one $x$-edge (resp.\ $y$-edge) of position~$j$ is contained in~$Y^*$, then it is easy to verify that removing all edges incident with the vertices of bit selector~$j$ and adding all $x$-edges (resp.\ $y$-edges) for that bit selector results in an induced matching which is not smaller, and in which the status of edges for other bit selectors is not changed; this follows from the fact that the adjacencies of the respective vertices to the outside the bit selector are identical. So in the remainder it suffices to consider a bit position~$j \in [\log r]$ for which~$Y^*$ contains no edge incident with a vertex in the bit selector. We exhibit an induced matching which is at least as large as~$Y^*$ and which has the desired form.

Observe that the bit selector for position~$j$ contains~$n$ triples, each of which forms a clique. As an induced matching cannot contain two edges incident with the same clique, $Y^*$ contains at most one edge incident with each triple for each bit selector~$j' \neq j$, and by assumption it contains no edges incident with bit selector~$j$. Since the union of the sets~$A_i$ for~$i \in [r]$ forms an independent set in~$G^*$, all matching edges in~$Y^*$ have at least one endpoint in~$B^*$, or one endpoint in a bit selector. As~$B^*$ has exactly~$n$ vertices, this bounds the number of edges in~$Y^*$ by~$n + ((\log r) - 1) n = n \log r$. Now observe that the union of all the $x$-edges of the bit selectors forms an induced matching of size~$n \log r$, and has the desired form. As we assumed~$Y^*$ to be maximum, the described induced matching is also maximum which concludes the proof.
\end{claimproof}

To complete the cross-composition it remains to prove that the constructed instance acts as the OR of the inputs. For the first direction, assume that~$G^*$ has a maximum induced matching~$Y^* \subseteq E(G^*)$ of size at least~$k^*$. By the second claim we may assume that for each~$j \in [\log r]$, the matching~$Y^*$ contains all the $x$-edges or all the $y$-edges of position~$j$. Now consider the instance number~$i^*$ whose binary expansion has a zero (resp.\ one) in the $j$-th bit position if~$Y^*$ contains the $x$-edges (resp.\ $y$-edges) of bit selector~$j$. By definition of the adjacencies of the bit selectors it follows that for all instance numbers~$i' \in [r]$ with~$i' \neq i^*$, no vertex of~$A_{i'}$ is the endpoint of an edge in~$Y^*$. To see this, consider a bit position~$j \in [\log r]$ where the binary expansion of~$i'$ and~$i^*$ differ; the $x$-vertices (resp.\ $y$-vertices) of instance selector~$j$ are endpoint of edges in~$Y^*$ whose other endpoints are formed by the $z$-vertices. As the $x$-vertices (resp.\ $y$-vertices) are adjacent to~$A_{i'}$ by the choice of~$j$, inducedness of the matching shows that~$A_{i'}$ contains no endpoints of matching edges. Hence~$Y^*$ is also an induced matching, of the same size, in the graph obtained from~$G^*$ by removing the vertices~$A_{i'}$ for~$i' \neq i^*$. Each triple of an instance selector is a clique, and by assumption on the form of~$Y^*$ the matching contains the $x$-edge or the $y$-edge of the triple. Since an induced matching cannot contain two edges incident with the same clique, this shows that no edges between~$A_{i^*}$ and an instance selector can be contained in~$Y^*$. Therefore it follows that if we delete the vertices~$A_{i'}$ for~$i' \neq i^*$ together with the vertices of the instance selectors from~$G^*$, we are left with an induced submatching of size at least~$k' - (n \log r) = k$. But the resulting graph is~$G^*[B^* \cup A_{i^*}]$, and as observed in the construction of~$G^*$ it is isomorphic to~$G_{i^*}$, which proves that~$G_{i^*}$ contains an induced matching of size~$k$ and is a \yes-instance.

For the reverse direction, assume there is some index~$i^* \in [r]$ such that~$G_{i^*}$ has an induced matching~$Y$ of size~$k$. As~$G^*[B^* \cup A_{i^*}]$ is isomorphic to~$G_{i^*}$, this implies that the induced subgraph admits an induced matching of size~$k$. Now augment this into an induced matching in~$G^*$ by adding the $x$-edges of the bit selectors for positions~$j$ where the binary expansion of~$i^*$ has a zero, and the $y$-edges where the expansion has a one. Using the description of~$G^*$ it is easy to verify that the resulting set of edges is an induced matching, containing a total of~$k + n \log r$ edges. This proves that~$(G^*, X^*, k^*)$ is a \yes-instance.

As the construction can be carried out in polynomial time and embeds the OR of the input instances into a single instance of the target problem with parameter value~$|X^*| = n + 3 n \log r$, this concludes the proof of \thmref{theorem:inducedMatchingByVCNoPoly}.
\myqed
\end{proof}

\subsection{Lower Bounds for Generalized Problem Statements}
As discussed in the introduction of Section~\ref{section:orderTesting} there are two obvious ways to attempt to generalize the positive results of Table~\ref{table:orderTesting}. We show that these generalizations for the induced subgraph testing problem (Section~\ref{subsection:test:subgraph:constant:vc}) and the minor testing problem (Section~\ref{subsection:test:minor:small:vc}) fail to admit polynomial kernels, unless \containment.

\subsubsection{Finding Induced Subgraphs with Constant-size Vertex Covers} \label{subsection:test:subgraph:constant:vc}
In this section we show that even the problem of testing for the existence of an induced subgraph with a constant-size vertex cover, is unlikely to admit a polynomial kernel when parameterized by the size of a vertex cover for the host graph. We use the following family of graphs for our proof.

\begin{definition}
Let~$s,t \geq 0$ be integers, and construct a graph as follows. Create a clique~$C_1$ on five vertices, and a vertex-disjoint clique~$C_2$ on four vertices. Add two vertices~$z_1$ and~$z_2$ and the edge~$z_1z_2$. Made~$z_1$ adjacent to all members of~$C_1$, and make~$z_2$ adjacent to all members of~$C_2$. Add~$s$ isolated vertices and make them adjacent to~$z_1$. Add~$t$ isolated vertices and make them adjacent to~$z_2$. The resulting graph is~$\Psi_{s,t}$.
\end{definition}

\noindent Observe that all graphs~$\Psi_{s,t}$ have a vertex cover of size~$11$ consisting of~$C_1 \cup C_2 \cup \{z_1,z_2\}$. We shall prove that the following problem is unlikely to admit a polynomial kernel, and thereby that the induced subgraph testing problem can still be hard to kernelize when looking for graphs with constant-size vertex covers.

\parproblemdef
{\InducedPsiTestByVC}
{A graph~$G$ with a vertex cover~$X$, and integer~$s,t \geq 0$.}
{The size~$|X|$ of the vertex cover.}
{Does~$G$ contain~$\Psi_{s,t}$ as an induced subgraph?}

\noindent We prove a superpolynomial kernel lower bound for this parameterized problem using cross-composition. The following variant of \IndependentSet will be used as the source problem for the composition.

\problemdef{\pTwoSplitIS}
{A graph~$G$, an independent set~$Y$ in~$G$ such that each component of~$G - Y$ is isomorphic to~$P_2$, and an integer~$k$.}
{Does~$G$ have an independent set of size at least~$k$?}

\noindent Jansen et al.~\cite[Lemma 10]{JansenB11} proved that \pTwoSplitIS is NP-complete, and used it to prove a kernel lower bound for a weighted version of \VertexCover. By adapting their construction, we prove a lower bound for \InducedPsiTestByVC.

\begin{theorem} \label{theorem:inducedpsitestbyvc:nopoly}
\InducedPsiTestByVC does not admit a polynomial kernel unless \containment.
\end{theorem}
\begin{proof}
By \thmref{crossCompositionNoKernel} and the NP-completeness of \pTwoSplitIS, it is sufficient to prove that \pTwoSplitIS cross-composes into \InducedPsiTestByVC. As in the cross-composition of Theorem~\ref{theorem:inducedMatchingByVCNoPoly}, we define a polynomial equivalence relation~\eqvr on instances of \pTwoSplitIS such that all malformed instances are equivalent. Two well-formed instances~$(G_1,Y_1,k_1)$ and~$(G_2,Y_2,k_2)$ are equivalent if~$k_1 = k_2$,~$|Y_1| = |Y_2|$ and~$|V(G_1)| = |V(G_2)|$. It is easy to verify that these choices satisfy Definition~\ref{polyEquivalenceRelation}.

We now give an algorithm that receives~$r$ instances of \pTwoSplitIS which are equivalent under \eqvr, and constructs an instance of \InducedPsiTestByVC with small parameter value that acts as the OR of the inputs. If the input instances are not well-formed, then we output a constant-sized \no-instance. From now on we may therefore assume that the input instances are~$(G_1, Y_1, k_1), \ldots, (G_r, Y_r, k_r)$ such that~$|V(G_1)| = \ldots = |V(G_r)| = n$, $|Y_1| = \ldots = |Y_r| = t$ and~$k_1 = \ldots = k_r = k$. As in the proof of Theorem~\ref{theorem:inducedMatchingByVCNoPoly} we may assume that~$r$ is a power of two. We construct an instance of \InducedPsiTestByVC as follows.

For each~$i \in [r]$, the graph~$G_i - Y_i$ contains~$n - t$ vertices and is a disjoint union of~$P_2$'s by the definition of \pTwoSplitIS. Let~$q=\frac{n-t}{2}$ be the number of~$P_2$'s in each graph~$G_i - Y_i$. For each~$i \in [r]$ label the vertices of the~$P_2$'s in~$G_i - Y_i$ by~$a_{i,1}, b_{i,1}, a_{i,2}, b_{i,2}, \ldots, a_{i,q}, b_{i,q}$ such that~$a_{i,j} b_{i,j}$ is an edge in~$G_i - Y_i$ for~$j \in [q]$; this implies that the only edges of~$G_i - Y_i$ are those between the $a$- and $b$-vertices with the same number. Construct a graph~$G^*$ as follows.
 
\begin{enumerate}
  \item Initialize~$G^*$ as the disjoint union of the input graphs~$G_1, \ldots, G_r$. This causes~$G^*$ to contain~$Y_i$ for all~$i \in [r]$.
  \item For each~$j \in [q]$, identify the vertices~$a_{1,j}, \ldots, a_{r,j}$ into a single vertex~$a^*_j$, and identify~$b_{1,j}, \ldots, b_{r,j}$ into a single vertex~$b^*_j$. Let~$D^* := \bigcup _{j \in [q]} \{a^*_j, b^*_j\}$. Observe that at this stage in the construction~$G^*[Y_i \cup D^*]$ is isomorphic to~$G_i$ for~$i \in [r]$.
	\item For~$j \in [\log r]$ add vertices~$s^0_j, s^1_j$ to~$G^*$, and add the edge~$s^0_j s^1_j$. Connect these to the remainder of the graph as follows.
	\begin{itemize}
		\item For~$i \in [r]$ and~$j \in [\log r]$, do the following. If the $j$-th bit of the binary expansion of number~$i$ is a zero, then make~$s^{0}_j$ adjacent to all vertices of~$Y_i$ that were added to~$G^*$ in the first step. If the bit is a one, then instead make~$s^1_j$ adjacent to~$Y_i$.
  \end{itemize}
\end{enumerate}

Before we continue the construction, let us observe that at this stage~$V(G^*)$ can be partitioned into three independent sets:~$\bigcup _{i\in[r]} Y_i$ is an independent set,~$(\bigcup _{j \in [q]} a^*_j) \cup (\bigcup _{j \in [\log r]} s^0_j)$ is an independent set, and the remainder~$(\bigcup _{j \in [q]} b^*_j) \cup (\bigcup _{j \in [\log r]} s^1_j)$ is an independent set. Hence~$G^*$ does not have a clique of size four or more at this point.

\begin{enumerate}[resume]
  \item Add a clique~$C_1$ on five vertices, and a clique~$C_2$ on four vertices, to~$G^*$.
  \item Add two vertices~$z_1, z_2$ and the edge~$z_1z_2$ to~$G^*$. Make~$z_1$ adjacent to~$C_1 \cup (\bigcup _{i \in [r]} Y_i) \cup D^*$, and make~$z_2$ adjacent to~$C_2 \cup (\bigcup _{j \in [\log r]} \{s^0_j, s^1_j\})$. This concludes the description of~$G^*$.
\end{enumerate}

Observe that as the edges between sets~$Y_i$ and~$D^*$ were not changed in these last steps, the final graph~$G^*[Y_i \cup D^*]$ is isomorphic to~$G_i$ for all~$i \in [r]$. Since~$G^*$ did not have  cliques of size four or more in its intermediate stage, it is easy to see that the unique maximum clique in~$G^*$ is~$C_1 \cup \{z_1\}$, consisting of six vertices. In the graph~$G^* - (C_1 \cup \{z_1\})$, the unique maximum clique is~$C_2 \cup \{z_2\}$ consisting of five vertices. We use this property of~$G^*$ in the proof of the following claim.

\begin{claim}
There is an index~$i \in [r]$ such that~$G_i$ has an independent set of size~$k$ if and only if~$G^*$ contains~$\Psi _{k, \log r}$ as an induced subgraph.
\end{claim}
\begin{claimproof}
($\Rightarrow$) Assume that~$G_{i^*}$ has an independent set of size~$k$ for~$i^* \in [r]$. Since~$G_{i^*}$ is isomorphic to~$G^*[Y_{i^*} \cup D^*]$, there is a size-$k$ independent set~$S^* \subseteq Y_{i^*} \cup D^*$ in~$G^*$. Consider the binary expansion of the number~$i^*$. Construct a vertex set~$B^*$ corresponding to this number as follows. For~$j \in [\log r]$, if the $j$-th bit of~$i^*$ is a one, then add~$s^0_j$ to~$B^*$. Otherwise add~$s^1_j$ to~$B^*$. We end up with a set~$B^*$ of size~$\log r$. Using the construction of~$G^*$ it is easy to see that~$B^*$ is independent in~$G^*$. Since we have picked the vertices corresponding exactly to the complement of the binary expansion of~$i^*$, there are no edges between~$S^*$ and~$B^*$. Now observe that by construction, ~$z_1$ is adjacent to all members of~$S^*$ but none of~$B^*$, whereas~$z_2$ is adjacent to all members of~$B^*$ but none of~$S^*$. Vertex~$z_1$ is adjacent to the five-clique~$C_1$, but no other vertices are adjacent to that clique, while~$z_2$ is the only vertex not in~$C_2$ that is adjacent to the four-clique~$C_2$. Since the edge~$z_1z_2$ is present,~$|S^*| = k$, and~$|B^*| = \log r$ it follows that~$G^*[S^* \cup B^* \cup C_1 \cup C_2 \cup \{z_1, z_2\}]$ is isomorphic to~$\Psi _{k, \log r}$, proving this direction of the claim.

($\Leftarrow$) Suppose that~$G^*$ contains~$\Psi _{k, \log r}$ as an induced subgraph. As~$G^*$ has a unique six-clique, and~$\Psi_{k, \log r}$ has a unique six-clique, these six-cliques must be mapped to each other by the isomorphism. Moreover, since~$z_1$ is the only vertex of the six-clique that has neighbors outside the six-clique (in both~$G^*$ and~$\Psi _{k, \log r}$), the vertices labeled~$z_1$ in~$G^*$ and~$\Psi _{k, \log r}$ must be mapped to each other by the induced subgraph isomorphism. Since the graph~$G^* - (C_1 \cup \{z_1\})$ has a unique five-clique, and~$\Psi_{k, \log r} - (C_1 \cup \{z_1\})$ has also a unique five-clique, we infer that these five-cliques must be mapped to each other. Again, since $z_2$ is the only vertex of the five-clique that has a neighbour outside it (in both~$G^*$ and~$\Psi _{k, \log r}$), the two copies of~$z_2$ must be mapped to each other by the isomorphism. Since the only neighbors that~$z_1$ has in~$G^*$ are~$C_1$,~$z_2$ and the set~$(\bigcup _{i\in[r]} Y_i) \cup D^*$, the vertices making up the size-$k$ side of~$\Psi_{k, \log r}$ must correspond to vertices of~$(\bigcup _{i\in[r]} Y_i) \cup D^*$ in~$G^*$. Let~$S^*$ be the~$k$ vertices in~$G^*$ that realize this size-$k$ side. Now consider the vertices in~$G^*$ that realize the~$\log r$-size side of~$\Psi_{k, \log r}$. Since the only neighbors of~$z_2$ in~$G^*$ are~$z_1$,~$C_2$, and~$\bigcup _{j \in [\log r]} \{s^0_j, s^1_j\}$ it follows that the size-$\log r$ side of~$\Psi_{k, \log r}$ is realized by vertices from~$\bigcup _{j \in [\log r]} \{s^0_j, s^1_j\}$; call these vertices~$U$. For each~$j \in [\log r]$ there is an edge~$s^0_j s^1_j$ by construction of~$G^*$. As the size-$\log r$ side of~$\Psi_{k, \log r}$ is an independent set,~$U$ contains at most one vertex of each such pair. As there are~$\log r$ pairs,~$U$ contains exactly one vertex of each pair. Define a number~$i^*$ as follows. For~$j \in [\log r]$, if~$s^0_j \in U$, let the $j$-th bit be a one; if~$s^1_j \in U$, let the $j$-th bit be a zero. Hence the number~$i^*$ is the complement of the binary string represented by the values encoded by~$U$, and therefore no vertex in~$U$ is adjacent to a vertex in~$Y_{i^*}$, by construction. For each~$i \in [r] \setminus \{i^*\}$, however, there is a bit position where the binary expansion of~$i$ differs with that of~$i^*$, and~$Y_i$ is adjacent to the vertex in~$U$ corresponding to that bit position. As there are no edges between the size-$k$ side and the size-$\log r$ side of~$\Psi_{k, \log r}$, the induced subgraph in~$G^*$ cannot contain vertices of~$\bigcup _{i \in [r] \setminus i^*} Y_i$. Hence the set~$S^*$ containing the~$k$ vertices that realize the size-$k$ side, is contained in~$Y_{i^*} \cup D^*$. But as~$G^*[Y_{i^*} \cup D^*]$ is isomorphic to~$G_{i^*}$, we find that~$S^*$ corresponds to a size-$k$ independent set in~$G^*[Y_{i^*} \cup D^*]$. Hence~$G_{i^*}$ has an independent set of size~$k$, concluding the proof.
\end{claimproof}

To define an instance of \InducedPsiTestByVC, observe that the set~$X^* := D^* \cup C_1 \cup C_2 \cup \{z_1, z_2\} \cup (\bigcup _{j \in [\log r]} \{s^0_j, s^1_j\})$ is a vertex cover in~$G^*$, since its complement consists of disjoint unions of independent sets. It is easy to verify that the size of~$X^*$ is polynomial in~$q + \log r$, which is polynomial in the encoding size of an input instance plus~$\log r$. The claim shows that the instance~$(G^*, X^*, s^* := k, t^* := \log r)$ is equivalent to the OR of the input instances. Since the construction can be carried out in polynomial time this is a valid cross-composition, and by \thmref{crossCompositionNoKernel} this concludes the proof.
\myqed
\end{proof}

\subsubsection{Finding Small Graphs as Minors} \label{subsection:test:minor:small:vc}
In this section we consider the minor testing problem parameterized by the sum of the vertex cover size and the size of the query graph.

\parproblemdef
{\HMinorTestByVCH}
{A graph~$G$ with a vertex cover~$X$, and a graph~$H$.}
{The value~$|X| + |V(H)|$.}
{Does~$G$ contain~$H$ as a minor?}

\noindent We prove a superpolynomial kernel lower bound for this problem using the technique of \emph{polynomial parameter transformations}, rather than cross-composition, since this simplifies the proof considerably. We therefore need the following terminology and results. For a parameterized problem~$Q \subseteq \Sigma^* \times \mathbb{N}$, the \emph{unparameterized version of~$Q$} is the set~$\tilde{Q} = \{x1^k \mid (x,k) \in Q\}$, where~$1$ is a new symbol that is added to the alphabet.

\begin{definition}[\cite{BodlaenderTY11}] \label{def:polyParamTransform}
Let~$P$ and~$Q$ be parameterized problems. We say that~$P$ is \emph{polynomial parameter reducible} to~$Q$, written $P \leq_{\mathrm{ptp}} Q$, if there exists a polynomial time computable function $g: \Sigma^* \times \mathbb{N} \to \Sigma^* \times \mathbb{N}$ and a polynomial~$p$, such that for all $(x,k) \in \Sigma^* \times \mathbb{N}$ we have (a) $(x,k) \in P \Leftrightarrow (x',k') = g(x,k) \in Q$ and (b) $k' \leq p(k)$. The function~$g$ is called \emph{polynomial parameter transformation}.
\end{definition}

\begin{theorem}[\cite{BodlaenderTY11}] \label{thm:kernelThroughReduction}
Let~$P$ and~$Q$ be parameterized problems and $\tilde{P}$ and $\tilde{Q}$ be the unparameterized versions of~$P$ and~$Q$ respectively. Suppose that $\tilde{P}$ is NP-hard and $\tilde{Q}$ is in NP. If there is a polynomial parameter transformation from~$P$ to~$Q$ and~$Q$ has a polynomial kernel, then~$P$ also has a polynomial kernel.
\end{theorem}

\noindent The contrapositive of Theorem~\ref{thm:kernelThroughReduction} can be used to obtain kernel lower bounds. We use the following problem as the starting point for the polynomial parameter transformation.

\parproblemdef
{\BipartiteRegularPerfectCodeByTK}
{A bipartite graph~$G$ with partite sets~$T$ and~$N$ such that all vertices in~$N$ have the same degree, and an integer~$k$.}
{$|T| + k$.}
{Is there a set~$N' \subseteq N$ of size at most~$k$ such that every vertex in~$T$ has exactly one neighbor in~$N'$?}

\noindent A set~$N'$ as described above is a \emph{perfect code} for~$G$.

\begin{lemma} [{\cite[Theorem 4]{DomLS09}}] \label{lemma:perfect:code:no:poly}
\BipartiteRegularPerfectCodeByTK does not have a polynomial kernel unless \containment.
\end{lemma}

\begin{theorem}
\HMinorTestByVCH does not admit a polynomial kernel unless \containment. \label{theorem:hminortestbyvch:nopoly}
\end{theorem}
\begin{proof}
We give a polynomial-parameter transformation from \BipartiteRegularPerfectCodeByTK to \HMinorTestByVCH. As the unparameterized version of the latter problem is easily seen to lie in NP, and the unparameterized version of the perfect code problem is NP-complete (it contains the NP-complete~\cite[SP2]{GareyJ79} \XTC problem as a special case), this suffices to prove the claim by Lemma~\ref{lemma:perfect:code:no:poly} and Theorem~\ref{thm:kernelThroughReduction}. So consider an instance~$(G,T,N,k)$ of \BipartiteRegularPerfectCodeByTK, and let~$r$ be the degree of vertices in~$N$. If~$k < |T| / r$, or~$|T| / r$ is not an integer, then we may safely output \no: at most~$k$ vertices of degree~$r$ cannot uniquely cover all~$|T|$ terminals. In the remainder, let~$k' := |T| / r$ be an integer; any perfect code for~$G$ of size at most~$k$, must have size exactly~$k'$. If~$r \geq |T| - 1$ then a perfect code consists of at most two elements of~$N$; we solve the problem in polynomial time and give the appropriate answer. We therefore assume that~$r < |T| - 1$ from now on. We create an instance of \HMinorTestByVCH consisting of a host graph~$G'$ and a query graph~$H'$.

Construct a graph~$G'$ from~$G$ by turning~$T$ into a clique; the vertex set of~$G'$ is~$T \cup N$. Construct a graph~$H'$ as follows. Start with a clique consisting of vertices~$v_{i,j}$ for~$i \in [k']$ and~$j \in [r]$. We use~$C$ to denote this clique. For~$i \in [k']$ add a vertex~$u_i$ adjacent to~$v_{i,1}, \ldots, v_{i,r}$. Denote these vertices by~$D$.

\begin{claim}
Graph~$G$ has a perfect code of size exactly~$k'$ if and only if~$G'$ contains~$H'$ as a minor.
\end{claim}
\begin{claimproof}
($\Rightarrow$) Assume that~$G$ has a perfect code~$N' \subseteq N$ of size exactly~$k'$. We claim that the subgraph of~$G'$ induced by~$N' \cup T$ is isomorphic to~$H'$. To prove this, we give an isomorphism~$f \colon N' \cup T \to V(H')$ such that for all~$u,v \in N' \cup T$ we have~$uv \in E(G[N' \cup T])$ if and only if~$f(u)f(v) \in E(H')$. Number the~$k'$ vertices in~$N'$ arbitrarily as~$n_1, \ldots, n_{k'}$. For~$i \in [k']$ define~$f(n_i) := u_i$, and consider the~$\deg_G(n_i) = r$ vertices~$N_G(n_i)$. Order them arbitrarily, mapping the first one to~$v_{i,1}$, the second one to~$v_{i,2}$, up to~$v_{i,r}$, under the isomorphism~$f$. Since~$N'$ is a perfect code, every vertex of~$N' \cup T$ is mapped to a unique vertex of~$H'$ by this choice of~$f$. It is straight-forward to verify the correspondence between edges of~$G'[N' \cup T]$ and edges of~$H'$. As an induced subgraph is a special case of a minor, this yields the proof in this direction.

($\Leftarrow$) Assume that~$G'$ contains~$H'$ as a minor, and let~$\phi$ be a minor model that maps~$V(H')$ to connected subsets of~$V(G')$. Consider an arbitrary vertex~$c \in C$ of~$H'$. As~$c$ is adjacent to all~$|C| - 1$ other members of the clique~$C$ in~$H'$, its degree in~$H'$ is at least~$|T|-1$. Since a vertex in~$N$ has degree~$r < |T|-1$, a branch set~$\phi(c)$ for~$c \in C$ cannot consist of a single vertex in~$N$, as such a vartex alone cannot be connected to~$|T|-1$ other branch sets. Since the vertices~$N$ are independent in~$G'$, and a branch set induced a connected subgraph, this implies that each set~$\phi(c)$ contains a vertex in~$T$. As~$|T| = |C|$ this implies that each branch set~$\phi(c)$ for~$c \in C$ contains exactly one vertex of~$T$. But then we may restrict each branch set~$\phi(c)$ to~$\phi(c) \cap T$ without breaking the minor model of~$H'$: vertices of~$N$ that might belong to the branch set are not needed to connect to other branch sets, as all possible connections to~$C$ are already made in the clique~$C$, and vertices of~$N$ do not connect to other vertices of~$N$ since~$N$ is an independent set. So if there is a minor model of~$H'$ in~$G'$, then there is one where the branch set of each~$c \in C$ consists of a unique vertex in~$T$. As~$N$ is an independent set, this also shows that~$\phi(u_i)$ is a singleton for each~$i \in [k']$: to contain more vertices and still induce a connected subgraph, a branch set~$\phi(u_i)$ would have to contain a vertex of~$T$. So we may assume that all branch sets in the minor model~$\phi$ are singleton vertices in~$G'$.

For each vertex $u_i$ let $n_i$ be such that $\phi(u_i)=\{n_i\}$. Let $N'=\{n_i \mid i\in [k']\}$; since vertices $n_i$ are pairwise different, it follows that $|N'|=k'$. We claim that $N'$ is a perfect code in $G$. Since $N'$ has size $k'=|T|/r$ and every vertex of $N'$ has degree exactly $r$ in $G$, a simple degree-counting argument shows that it suffices to argue that each vertex of $T$ is adjacent to at least one vertex of $N'$. Consider any vertex $w\in T$. Since $f$ is surjective on $T$, there exist some indices $i,j$, where $i\in [k']$ and $j\in [r]$, such that $\{w\}=\phi(v_{i,j})$. The vertex $v_{i,j}$, however, is adjacent to $u_i$ in $H'$, so it follows that $w$ must be adjacent to $n_i$ in $G$. As $w$ was picked arbitrarily, we conclude that every vertex of $T$ is adjacent to at least one vertex of $N'$ and we are done.
\end{claimproof}

Observe that the set~$T$ forms a vertex cover of~$G'$. The tuple~$(G',H',X' := T)$ can therefore serve as an instance of \HMinorTestByVCH. As we established earlier that any perfect code in~$G$ must have size exactly~$k'$, the claim shows the equivalence between the original instance and the constructed instance. The new value of the parameter is~$|X'| + |V(H')| = |T| + (|T| + k') \leq 2|T| + k$, which is polynomial in the original parameter of the \BipartiteRegularPerfectCodeByTK instance. As the transformation can easily be computed in polynomial time, it is a polynomial-parameter transformation, which concludes the proof.
\myqed
\end{proof}

Concerning the minor-testing variant of the parameterization discussed in this section, note that the kernel lower bound for \InducedPathByVC (Theorem~\ref{theorem:inducedPathByVCNoPoly}) already implies that \HInducedSubgraphTestByVCH does not admit a polynomial kernel unless \containment.

\section{Conclusion}
We have studied the existence of polynomial kernels for graph problems parameterized by vertex cover. The general theorems we presented unify known positive results for many problems, and the characterization in terms of forbidden or desired induced subgraphs from a class characterized by few adjacencies gives a common explanation for the results obtained earlier. Our comparison of induced subgraph and minor testing problems shows that the kernelization complexity landscape of problems parameterized by vertex cover is rich and difficult to capture with a single meta-theorem. The kernel lower bounds for \VariableBicliqueTest, \InducedPathByVC, and \InducedMatchingByVC, show that besides connectivity and domination requirements, an inducedness requirement can form an obstacle to polynomial kernelizability for parameterizations by vertex cover.

An obvious direction for further work is to find even more general kernelization theorems that can also encompass the known positive results for problems like \TreewidthVC~\cite{BodlaenderJK11b}, \PathwidthVC~\cite{BodlaenderJK12a}, and \CliqueMinorTestVC. There are also various problems for which the kernelization complexity parameterized by vertex cover is still open; among these are \PerfectDeletion, \IntervalDeletion, \bandwidth, and \orientablegenus. One may also investigate whether \thmref{theorem:deletionToPiVC} has an analogue for edge-deletion problems.

In light of the parameter ecology program~\cite{FellowsJR12} it is natural to ask whether the general kernelization theorems obtained in \sectref{section:metatheorems} can be transferred to smaller parameters than the vertex cover number. As this parameter measures the vertex-deletion distance to a graph of treewidth zero, an obvious next step would be parameterization by the feedback vertex number --- the vertex-deletion distance to a graph of treewidth one. Unfortunately, this seems difficult. While \VertexCover and \OddCycleTransversal admit polynomial kernels for this parameter~\cite{JansenB11,JansenK12}, the kernelization schemes are rather involved and lack any similarity. For the \LongPath problem, the existence of a polynomial kernel parameterized by feedback vertex number is still open. In the case of \ThreeColoring~\cite{JansenK11b} and \disjointPaths~\cite{BodlaenderTY11} we even know that no polynomial kernel exists for the parameterization by feedback vertex number (unless \containment). Hence it seems that a better understanding of polynomial kernelizability for parameterizations by feedback vertex number is needed before attempting to capture the phenomenon by general theorems.

The case study of \sectref{section:orderTesting} raises some interesting questions. To devise a polynomial kernel for \ConstantBicliqueTest we used a reduction to~$|X|^{\Oh(1)}$ instances of a kernelizable problem. The guessing phase leading to the series of instances is reminiscent of a Turing kernelization (cf.~\cite{Binkele-RaibleFFLSV12,HermelinKSWW11}). Can the power of Turing kernelization be exploited to give polynomial kernels for induced subgraph problems that do not admit polynomial many-one kernels? For example, does the \VariableBicliqueTest problem admit a polynomial Turing kernel, or can the recent framework of Hermelin et al.~\cite{HermelinKSWW11} be used to prove that this is unlikely? The question of Turing kernelization seems especially relevant for the area of induced subgraph testing, as \Clique parameterized by vertex cover does not admit a polynomial many-one kernel (unless \containment) but has a trivial linear-vertex Turing kernel~\cite{BodlaenderJK11}. Could it be that the induced $H$-subgraph testing problem has a polynomial Turing kernel for any graph~$H$ as input, when parameterized by vertex cover?

\bibliography{Paper}

\end{document}